\documentclass[journal,twocolumn,11pt]{IEEEtran}
\onecolumn
\usepackage{setspace}
\usepackage{booktabs} 
\usepackage[utf8]{inputenc}
\usepackage[english]{babel}

\usepackage{bbding}
\usepackage{pifont}
\usepackage{ulem}
\usepackage{amsmath}
\usepackage{amssymb}

\usepackage{framed}
\usepackage{epsfig}
\usepackage{mathtools}
\usepackage{floatflt} 
\usepackage{paralist} 
\usepackage{algorithm} 
\usepackage{graphicx}
\usepackage{listings}    
\usepackage{xcolor}
\usepackage{hyperref}
\usepackage{indentfirst}
\usepackage{booktabs}
\usepackage{subfigure}
\usepackage{caption}
\usepackage{calligra}
\usepackage{bm}
\usepackage{fancyhdr}
\usepackage{float}
\usepackage{pifont}
\usepackage{colortbl} 
\usepackage{amsthm}
\usepackage{hyperref}

\def\f{\mathbf{f}}

\def\V{\mathcal{V}}
\def\E{\mathcal{E}}

\def\one{{\bf 1}}

\def\fw{\widehat{\f}}

\newcommand{\mypar}[1]{{\bf #1.}}

\newtheorem{myLem}{Lemma}

\DeclareMathOperator{\Y}{Y}

\DeclareMathOperator{\Adj}{A}

\DeclareMathOperator{\D}{D}

\DeclareMathOperator{\F}{F}

\DeclareMathOperator{\Pj}{P}
\DeclareMathOperator{\Q}{Q}

\DeclareMathOperator{\Ss}{S}
\DeclareMathOperator{\Vm}{V}

\DeclareMathOperator{\Um}{U}
\DeclareMathOperator{\Z}{Z}
\DeclareMathOperator{\W}{W}

\DeclareMathOperator{\M}{M}

\def\Fw{\widehat{\F}}

\newcommand{\R}{\ensuremath{\mathbb{R}}}
\DeclareMathOperator{\Id}{I}

\begin{document}


\title{Fast, Warped Graph Embedding:
\\ Unifying Framework and One-Click Algorithm}

\author{Siheng~Chen, Sufeng Niu,
Leman Akoglu, Jelena~Kova\v{c}evi\'c, Christos Faloutsos}
 
\maketitle

\begin{abstract}
What is the best way to describe a user in a social network with just a few numbers? Mathematically, this is equivalent to assigning a vector representation to each node in a graph, a process called graph embedding.  We propose a novel framework, GEM-D that unifies most of the past algorithms such as LapEigs, DeepWalk and node2vec. GEM-D achieves its goal by decomposing any graph embedding algorithm into three building blocks: node proximity function, warping function and loss function. Based on thorough analysis of GEM-D, we propose a novel algorithm, called UltimateWalk, which outperforms the most-recently proposed state-of-the-art DeepWalk and node2vec. The contributions of this work are: (1) The proposed framework, GEM-D unifies the past graph embedding algorithms and provides a general recipe of how to design a graph embedding; (2) the nonlinearlity in the warping function contributes significantly to the quality of embedding and the exponential function is empirically optimal; (3) the proposed algorithm, UltimateWalk is one-click (no user-defined parameters), scalable and has a closed-form solution.

\end{abstract}

 


\begin{figure}[htb]
  \begin{center}
    \begin{tabular}{cc}
\includegraphics[width=0.18\columnwidth]{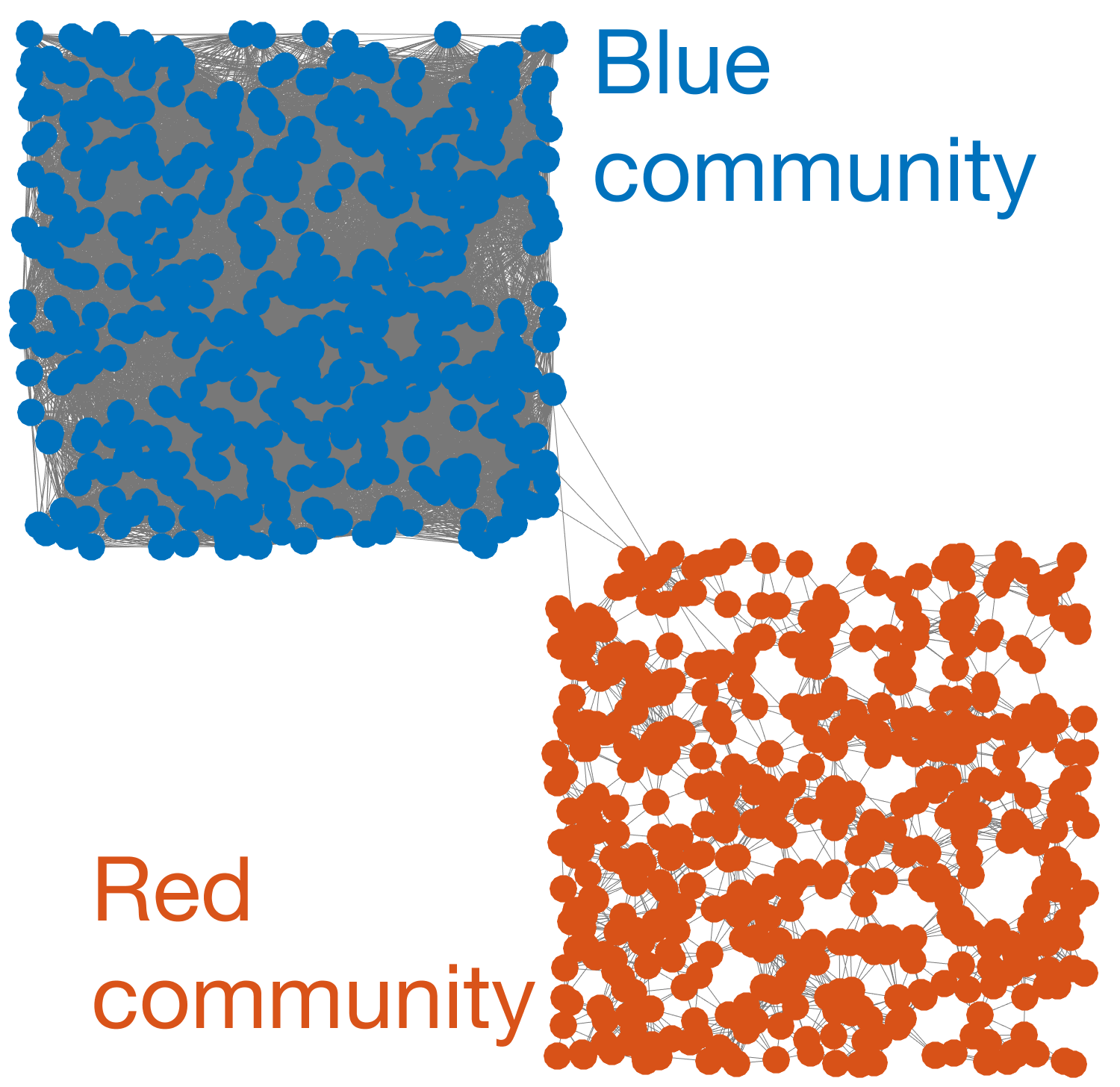}   &
\includegraphics[width=0.18\columnwidth]{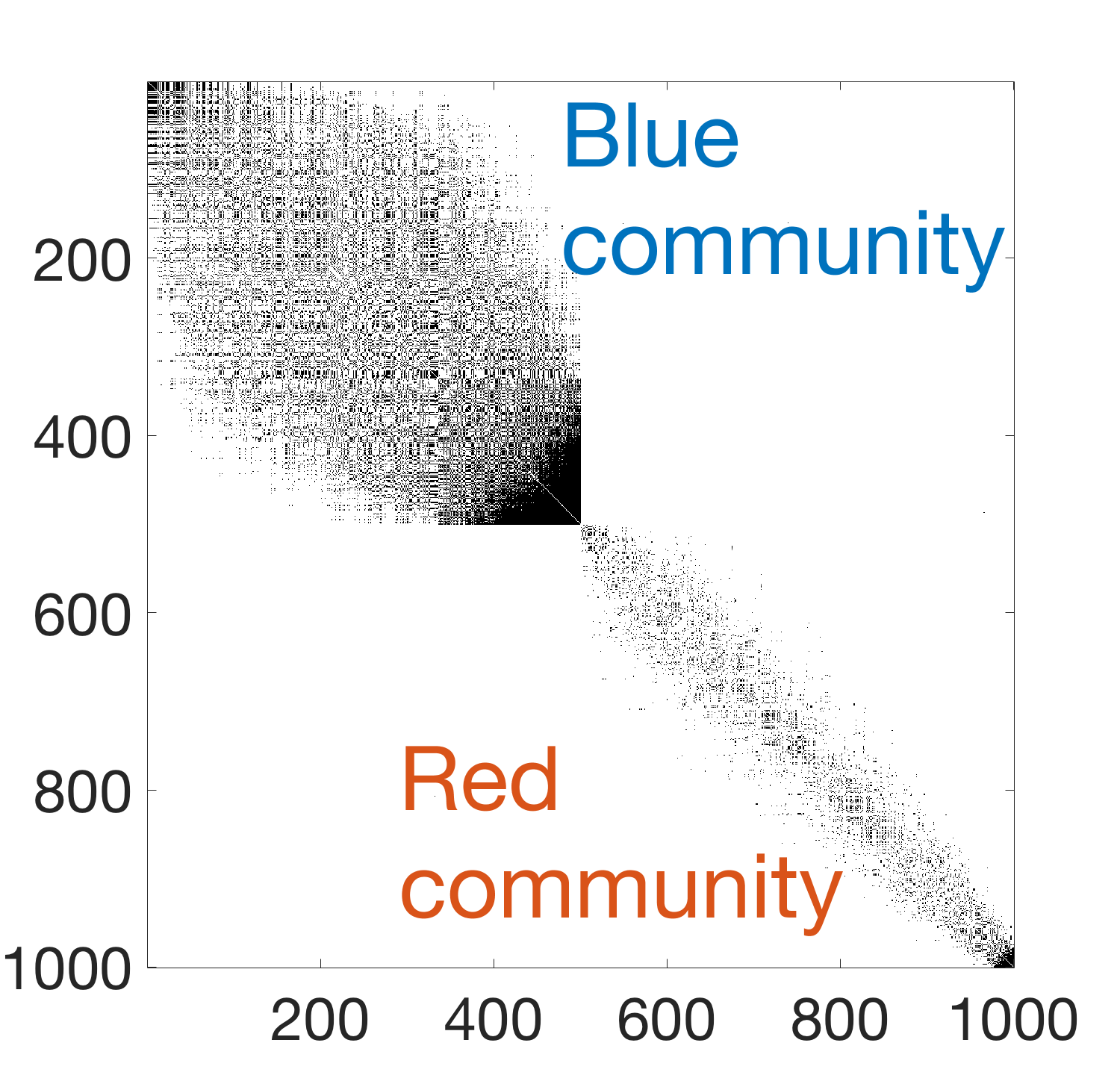}   
\\
 {\small (a)  Input graph. }  &  {\small (b)  Adjacency matrix.} 
\\
  \includegraphics[width=0.18\columnwidth]{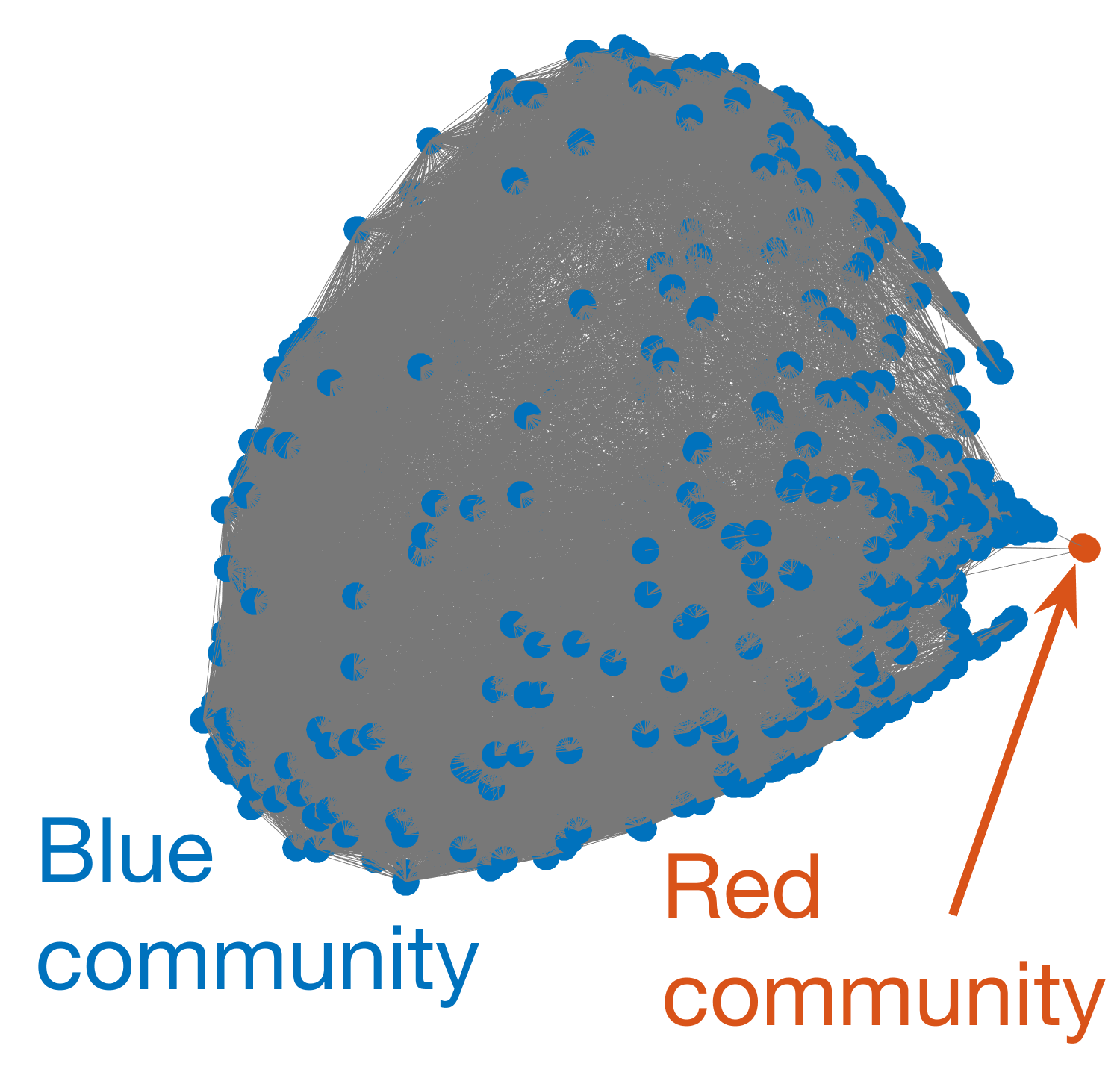}   & \includegraphics[width=0.18\columnwidth]{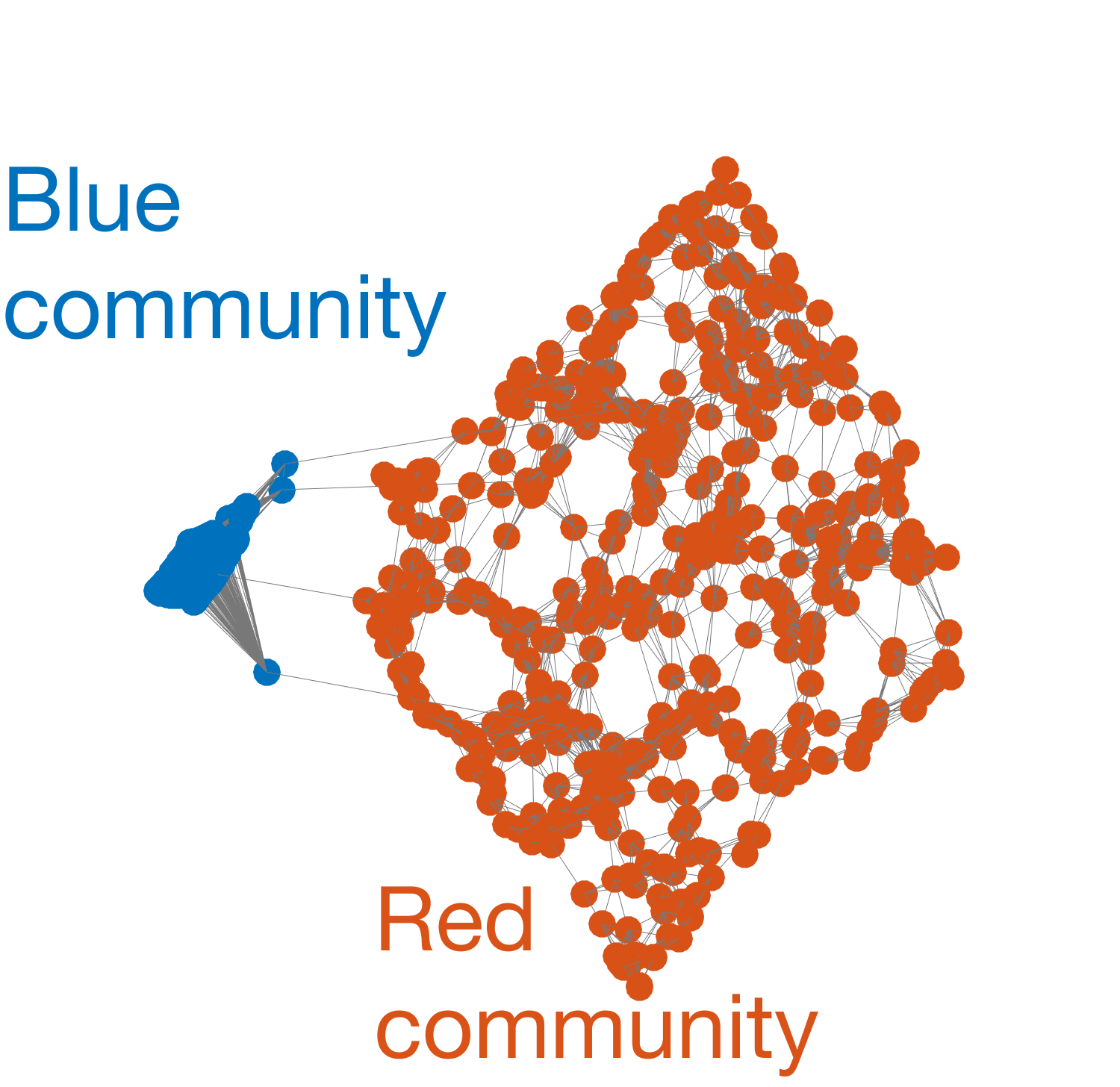} 
\\
 {\small (c) Linear embedding through SVD.} &  {\small (d)  UltimateWalk.}
    \end{tabular}
  \end{center}
  \caption{\label{fig:toy_example}  
The proposed algorithm UltimateWalk introduces a warping function, which significantly improves the quality of graph embedding. We expect the nodes in the blue community should be close in the graph embedding domain because their connections are strong, which is shown by UltimateWalk.   }
\end{figure}

\begin{figure}[htb]
  \begin{center}
    \begin{tabular}{cc}
\includegraphics[width=0.3\columnwidth]{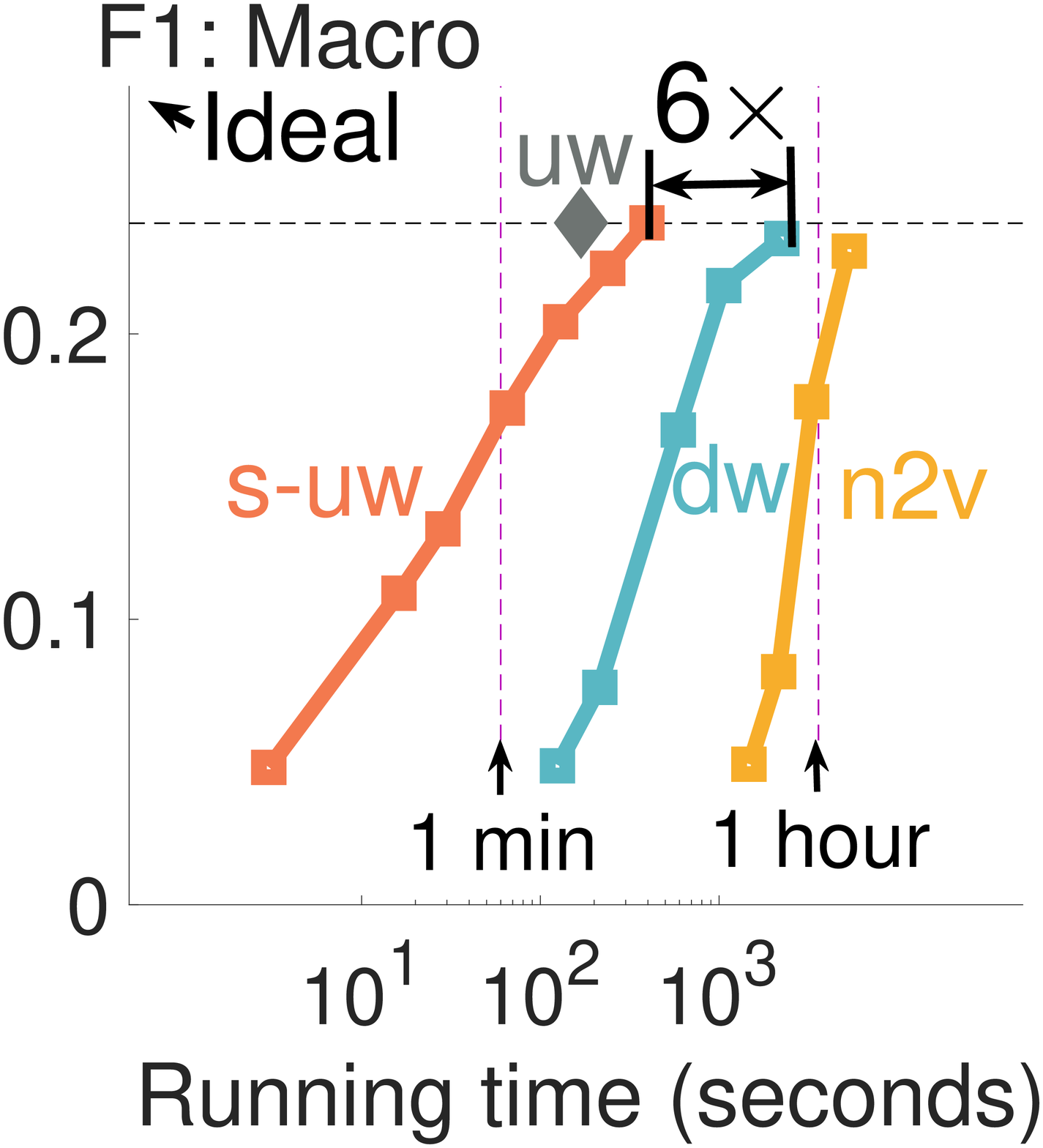}   & \includegraphics[width=0.3\columnwidth]{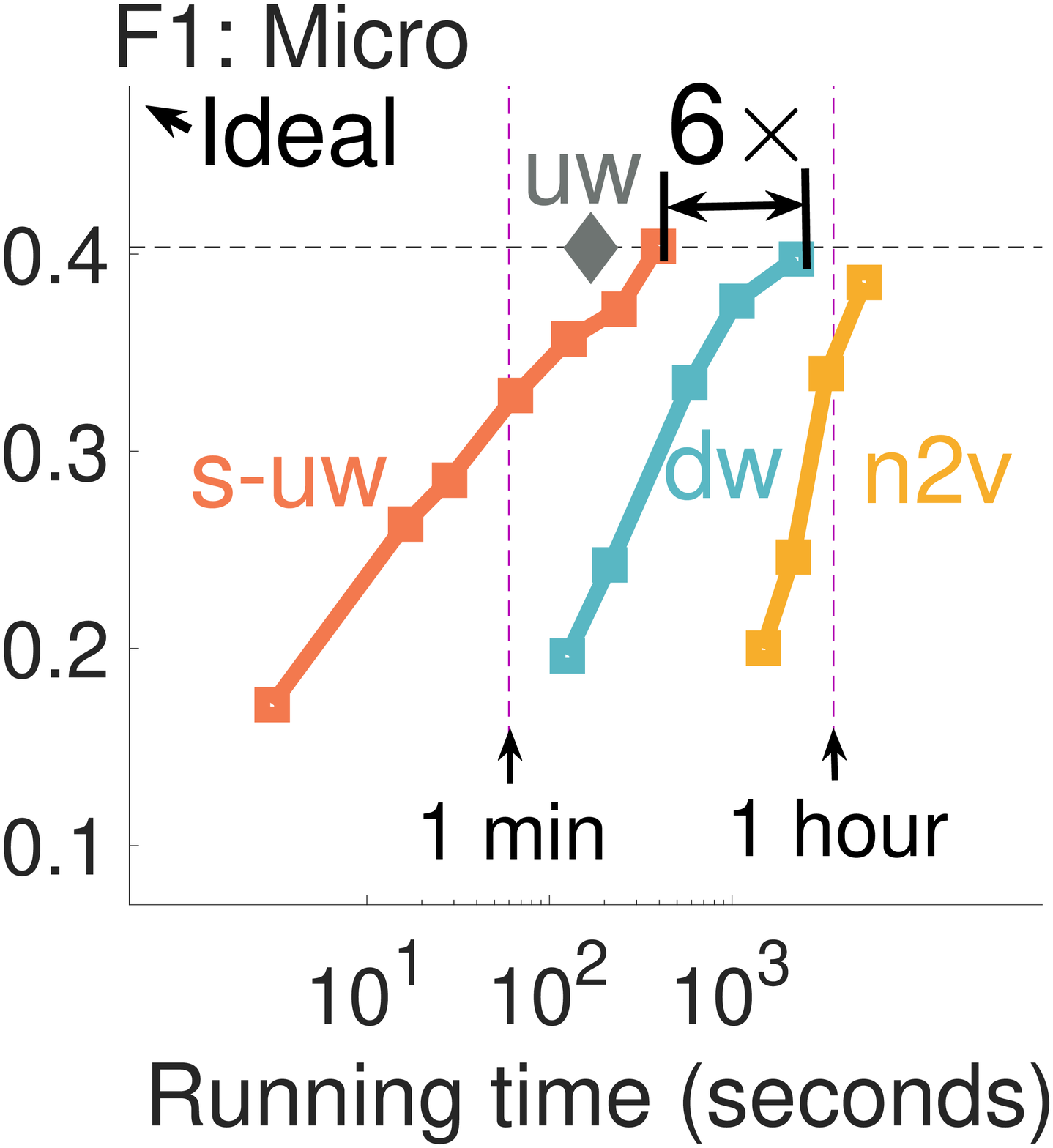}
\\
 {\small (a) Macro in BlogCatalog. }  &   {\small (b) Micro  in BlogCatalog. }
    \end{tabular}
  \end{center}
  \caption{\label{fig:crown} \emph{UltimateWalk outperforms competition.} The proposed algorithm UltimateWalk (uw) and its scalable variant (s-uw) outperforms the state-of-the-art DeepWalk (dw) and node2vec (n2v). }
\end{figure}

\begin{figure}[thb]
  \begin{center}
\includegraphics[width=0.5\columnwidth]{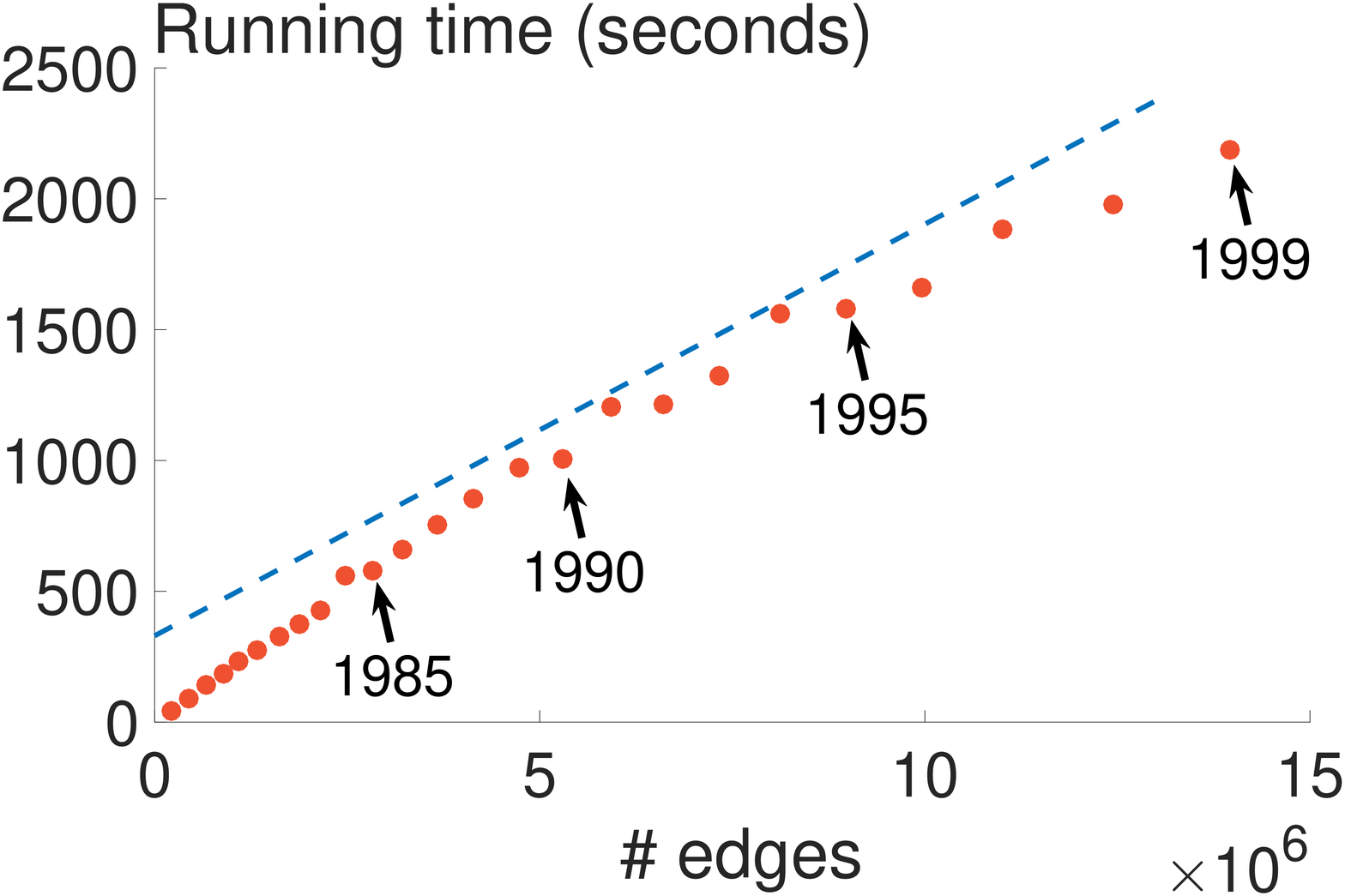}  \\
  \end{center}
  \caption{\label{fig:patent} \emph{UltimateWalk scales linearly.} The proposed algorithm UltimateWalk scales well, linearly with the number of edges. (U.S. patent 1975--1999 dataset; see Lemma{\color{red}~\ref{thm:complexity}}.)  }
\end{figure}

\section{Introduction}
A graph embedding is a collection of feature vectors associated with nodes in a graph; each feature vector describes the overall role of the corresponding node. Through a graph embedding, we are able to visualize a graph in a 2D/3D space and transform problems from a non-Euclidean space to a Euclidean space, where numerous machine learning and data mining tools can be applied. The applications of graph embedding include graph visualization~\cite{HermanMM:00}, graph clustering~\cite{Luxburg:07}, node classification~\cite{SenNBGGE:08, ZhuGL:03}, link prediction~\cite{Liben-NowellK:07,BackstromL:11}, recommendation~\cite{YuRSGSKNH:14}, anomaly detection~\cite{AkogluTK:15}, and many others. A desirable graph embedding algorithm should (1) work for various types of graphs, including directed, undirected, weighted, unweighted and bipartite; (2) preserve the symmetry between the graph node domain and the graph embedding domain; that is, when two nodes are close/far away in the graph node domain, their corresponding embeddings are close/far away in the graph embedding domain; (3) be scalable to large graphs; (4) be interpretable; that is, we understand the role of each building block in the graph embedding algorithm.

\begin{table*}
  \footnotesize
  \begin{center}
   \begin{tabular}{@{}lllllll|ll@{}}
      \toprule
      {\bf GEM-D} &  {\bf LapEigs~\cite{BelkinN:03}} & {\bf DiffMaps~\cite{CoifmanL:06}} & {\bf MF~\cite{KorenBV:09}}  &    {\bf LINE~\cite{TangQWZYM:15}}  &   
       {\bf DeepWalk~\cite{PerozziAS:14}} &  {\bf node2vec~\cite{GroverL:16}} &  {\bf UltimateWalk }
        \\
      \midrule
 	   {\bf  Loss function $d(\cdot, \cdot)$} &    Frobenius   & Frobenius   &  Frobenius  & KL   & KL  & KL  & Frobenius \\
 	 {\bf  Warping function  $g(\cdot)$}  &    linear & linear & linear & sigmoid  & exponential & exponential & exponential  \\
    {\bf  Proximity function  $h(\cdot)$} &  Laplacian & transition  & adjacency  & transition  & FST & FSMT  & FST \\
    \midrule
    {\bf  Closed form} &  \CheckmarkBold & \CheckmarkBold & \CheckmarkBold  &   &  &   & \CheckmarkBold \\
    {\bf  High-order proximity}   &  & \CheckmarkBold   &  & \CheckmarkBold  & \CheckmarkBold & \CheckmarkBold  & \CheckmarkBold \\
    {\bf  Warped} &   &   &   & \CheckmarkBold   & \CheckmarkBold & \CheckmarkBold  & \CheckmarkBold \\
      \addlinespace[1mm] \bottomrule
    \end{tabular}
  \end{center}
  \caption{\label{table:general} Unifying power of GEM-D. Many graph embeddings, including Laplacian Eigenvectors (LapEigs), Diffusion Maps (DiffMaps), matrix factorization (MF),  LINE, DeepWalk, node2vec are specific cases.}
\end{table*}

Numerous embedding algorithms have been proposed in data mining, machine learning and signal processing communities. IsoMAP~\cite{TenenbaumDL:00}, LLE~\cite{RoweisS:00}, LapEigs~\cite{BelkinN:03}, diffusion maps~\cite{CoifmanL:06} estimate the intrinsic manifold based on the distribution of data points;  SNE~\cite{HintonR:02} and t-SNE~\cite{MaatenH:08}, LargeVis~\cite{TangLZM:16} visualize high-dimensional data in a 2D/3D scatter plot; 
word2vec~\cite{MikolovCCD:13,MikolovSCCD:13} learn feature vectors for words in documents; DeepWalk~\cite{PerozziAS:14}, LINE~\cite{TangQWZYM:15}, node2vec~\cite{GroverL:16}, HOPE~\cite{OuCPZ:16}  and structural deep network embedding~\cite{WangCZ:16}  learn feature vectors for nodes to handle prediction tasks in large-scale networks; and community-based graph embeddings~\cite{WangCZPZY:17, ZhengCCCC:16} further preserve macroscopic information, such as communities.


On the heels of these outstanding contributions, {\it can we find a
  framework that will unify and explain them, a recipe to design a
  good graph embedding and an answer to what leads to fundamental
  improvements to graph embedding algorithms?} To answer these questions,   we propose a unifying graph embedding framework, called~\emph{Graph EMbedding Demystified} (GEM-D), which decomposes a graph embedding algorithm into three building blocks: node proximity function, warping function and loss function.
GEM-D clarifies the functionality of each building block and subsumes many previous graph embedding algorithms as special cases. Based on the insights GEM-D brings, we propose a novel graph embedding algorithm,~\emph{UltimateWalk}, which inherits advantages of previous algorithms; it is general enough to handle diverse types of graphs, has a closed-form solution and is effective, scalable and interpretable.

The contributions of our paper are as follows:

$\bullet$ {\bf Unifying framework.} GEM-D provides a unifying framework to design a graph embedding, subsumes LapEigs, DeepWalk and node2vec as special cases, and provides insight into the asymptotic behaviors of DeepWalk and node2vec; see Table{\color{red}~\ref{table:general}} and Section{\color{red}~\ref{sec:closed}};

$\bullet$ {\bf Fundamental understanding.} By decomposing a graph embedding algorithm into three building blocks, we are able to thoroughly analyze each and find that (1) nonlinearity introduced in the warping function is the key contributor to great performance; (2) the distribution of the elements in the proximity matrix is an important indicator of overall performance and (3) the memory factor only slightly influences the overall performance (see Figures{\color{red}~\ref{fig:blog_nonlinear}} and{\color{red}~\ref{fig:kaggle_1968_node2vec}}).

$\bullet$ {\bf One-click algorithm.}  UltimateWalk is effective and without user-defined parameters; it outperforms the  state-of-the-art algorithms (see Figures{\color{red}~\ref{fig:crown}} and{\color{red}~\ref{fig:blog_converge}}) and scales linearly with the number of input edges (see Figure{\color{red}~\ref{fig:patent}}).


{\bf Reproducibility} The datasets are available online; see Section{\color{red}~\ref{sec:exp}}. Our code is available at~\url{https://users.ece.cmu.edu/~sihengc/publications.html}.

\section{Framework: Graph EMbedding Demystified}
\label{sec:proposed}
Based on the existing graph embedding algorithms, such as Laplacian
eigenvectors, DeepWalk and node2vec, can we find a framework that will
unify and explain them all, a recipe to design a good graph
embedding and an answer to what leads to fundamental improvements to
graph embedding algorithms? We propose such a framework in this
section.

\subsection{General formulation}
Let $\mathcal{G} = (\V,\E, \Adj)$  be a directed and weighted graph, where $\V = \{v_i \}_{i=1}^N$ is the set of nodes, $\E$ is the set of weighted edges and $\Adj \in \R^{N \times N}$ is the weighted adjacency matrix\footnote{The proposed framework and algorithm can be easily extended to bipartite graphs.}, whose element $\Adj_{i,j}$ measures the direct relation between the $i$th and  $j$th nodes.  Let $\D \in \R^{N \times N}$ be a diagonal degree matrix, where $\D_{i,i} = \sum_{j \in \V} \Adj_{i,j}$.  To represent the pairwise node proximity in the graph node domain, we define a proximity matrix as $\Pi = h(\Adj) \in \R^{N \times N}$, where $\Adj$ is the adjacency matrix and $h(\cdot)$ is a node proximity function we discuss in Section{\color{red}~\ref{sec:prox}}. 

Let $\F, \Fw \in \R^{N \times K} (K \ll N)$ be a pair of dual embedding matrices with their $i$th rows, $\f_i, \fw_i \in \R^K$, the hidden features of the $i$th node. $\F$ and $\Fw$ capture in-edge and out-edge behaviors, respectively. The final graph embedding is a concatenation of these two.  In the graph embedding domain, we use the inner product of $\f_i$ and $\fw_j$ to capture the pairwise proximity between the $i$th node and  $j$th nodes, that is, we define $g(\F \Fw^T ) \in \R^{N \times N}$ as the proximity matrix in the graph embedding domain. The $(i,j)$th element $g(\f_i^T \fw_j )$ measures the proximity between the $i$th  and $j$th nodes in the graph embedding domain, with the warping function $g(\cdot)$  a predesigned element-wise, monotonically increasing function, which we discuss in Section{\color{red}~\ref{sec:warp}}.

Our goal is to find a pair of low-dimensional graph embeddings such that the node proximities in both domains are close, that is, $\Pi_{i,j}$ is close to $g(\f_i^T \fw_j)$ for each pair of $i,j$. We thus consider a general formulation, GEM-D, as
\begin{eqnarray}
  \label{eq:gem}
  \F^{*}, \Fw^{*} 
  & = & \arg \min_{\F, \Fw \in \R^{N \times K}}  d ( h(\Adj), g(\F \Fw^T) ),
\end{eqnarray}
where $d(\cdot, \cdot)$ is a loss function that measures the difference; we  discuss it in Section{\color{red}~\ref{sec:loss}}.  $g(\F \Fw^T)$ is generated from a low-dimensional space as $\F^{*}, \Fw^{*}  \in \R^{N \times K}$, which is used to approximate an input $\Pi$. We minimize the loss function{\color{red}~\eqref{eq:gem}} to control the approximation error.  

GEM-D involves the following three important building blocks, denoted by GEM-D$[h(\cdot), g(\cdot), d(\cdot, \cdot)]$:
\begin{itemize}
\item the node proximity function $h(\cdot): \R^{N \times N}
  \rightarrow \R^{N \times N}$;
\item the warping function $g(\cdot) : \R \rightarrow \R$ (when
  $g(\cdot)$ applies to a matrix, it is element-wise); and
\item the loss function $d(\cdot, \cdot) : \R^{N \times N} \times
  \R^{N \times N} \rightarrow \R$.
\end{itemize}

The general framework GEM-D makes it clear that all information comes from the graph adjacency matrix $\Adj$. Based on the properties of graphs, we adaptively choose three building blocks and solve the corresponding optimization problem to obtain a desirable graph embedding.
We consider GEM-D at a general and abstract level to demystify the constitution of a graph embedding algorithm and analyze the function of each  building block.  In what follows, we discuss each building block in detail and show that GEM-D is a unifying framework that concisely summarizes many previous graph embedding algorithms.  

\subsection{Node proximity function $h(\cdot)$}
\label{sec:prox}
The proximity matrix is a function of the adjacency matrix, $\Pi = h(\Adj)$; in other words, all information comes from given pairwise connections. The node proximity function extracts useful features and removes useless noise for subsequent analysis; for example,  the node proximity function can normalize the effect from hubs and capture the high-order proximity information. Node proximity function is implicitly used in DeepWalk and node2vec, where the proximity matrix is obtained via random-walk simulation. We consider an explicit matrix form, so that we are able to analyze the asymptotic behavior. We now consider some choices for the proximity matrix $\Pi$. 

\subsubsection{Laplacian matrix} The Laplacian matrix $\L \in \R^{N \times N}$ is defined as $\L = \D - \Adj$; this is the matrix used in LapEigs~\cite{BelkinN:03}.

\subsubsection{Transition matrix} The transition matrix $\Pj \in \R^{N \times N}$ is defined as $\Pj = \D^{-1} \Adj$. The  element $\Pj_{i,j}$ is the transition probability of a random walker at the $i$th node visiting the $j$th node in a single step; this is the matrix used in diffusion maps~\cite{CoifmanL:06} and LINE~\cite{TangQWZYM:15}.

Both the Laplacian matrix and the transition matrix normalize the original adjacency matrix based on the node degrees, which reduces the effect from hubs.

\subsubsection{Finite-step transition matrix}
The finite-step transition (FST) matrix is defined as a finite-order polynomial in $\Pj$ (transition matrix),
$
\Pi^{(L)} \ = \ \sum_{\ell=1}^L \Pj^{\ell},
$
where  $L \geq 1$ is a finite number of steps. The FST matrix models the  random walker walking $L$ steps randomly and sequentially. When a random walker starts at the $i$th node,  the expected number of times to visit the $j$th node is $\Pi^{(L)}_{i,j}$; this is the matrix implicitly used  in DeepWalk~\cite{PerozziAS:14}. 


\subsubsection{Infinite-step transition  matrix}
The infinite-step transition  (IST) matrix is defined as an infinite-order polynomial in $\Pj$ (transition matrix),
\begin{equation*}
  \label{eq:infinite_step_memory_free}
  \Pi^{(\alpha)}  \ = \ \sum_{\ell=1}^{\infty} \alpha^{\ell-1} \Pj^{\ell}
  \ = \ \frac{1}{\alpha}   \left( \mathcal{P}  - \Id  \right),
\end{equation*}
where $\mathcal{P} = \left( \Id - \alpha \Pj \right)^{-1}$ and $0 < \alpha < 1$. The IST matrix models the  random walker stopping after each step with probability $1-\alpha$. When a random walker starts at the $i$th node,  the expected number of times to visit the $j$th node is $\Pi^{(\alpha)}_{i,j}$; this is the matrix used in PageRank~\cite{BrinP:98} and personalized PageRank~\cite{Haveliwala03}.

Both FST and IST matrices consider not only direct relations but high-order, long-range node proximities.



\subsubsection{Finite-step memory-modulated transition matrix}
The finite-step memory-modulated transition (FSMT) matrix is a generalized version of the FST matrix, which also takes the memory into account.  The memory factor is introduced in node2vec~\cite{GroverL:16} without an explicit matrix form. Here we derive its closed form. The memory factor records each previous step and influences the walking direction, leading to a biased random walk, which is a trade-off between breadth-first search (BFS) and depth-first search (DFS). The memory matrix $\M \in \R^{N \times N}$ is defined as
\begin{equation*}
 \M_{ i, k }  = 
  \left\{ 
    \begin{array}{rl}
      1/p, & i = k; \\
      1, &  i~{\rm and}~k~{\rm are~adjacent}; \\
      1/q, & {\rm geodesic~distance~between}~i~{\rm and}~k~{\rm is}~2; \\
      0, & \mbox{otherwise}.
  \end{array} \right.
\end{equation*}
The memory factors $p, q$ control the walking direction: when $p$ is small, the walker tends to do BFS and when $q$ is small, the walker tends to do DFS. The FSMT matrix is defined as



\begin{equation*}
  \label{eq:finite_step_memory}
  \Pi^{(L, p, q)} = \Phi \left( \sum_{\ell=0}^{L-1}  \W^{\ell} \right) \Q^T,
\end{equation*}
where $L \geq 1$ is a finite number of walk length and the initial probability matrix is
$
 \Q  =
\begin{bmatrix}
\Q_{1}
&
\Q_{2}
&
\cdots
&
\Q_{N}
\end{bmatrix}
  \in~\R^{N \times N^2},
$
where $\Q_{k}   \in \R^{N \times N}$ with
$\left(  \Q_{k}  \right)_{m, \ell} \ = \
\Pj_{  k, m}  \mathcal{I}{( \ell = m )}$ ($\Pj$ is the transition matrix and $\mathcal{I}{( \cdot )}$ is the indicator function),
the merging matrix is
$
\Phi \ = \ \Id_N \otimes  \one_N \in \R^{N \times N^2},
$
where $\Id_N \in \R^{N \times N}$ is an identity matrix, $\one_N \in \R^N$ is an all-one vector, and
the expanded transition matrix is
\begin{eqnarray*}
 \W  \ = \
\begin{bmatrix}
\W_{1,1}  & \W_{1,2}  & \cdots   &  \W_{1,N} 
\\
\W_{2,1}  & \W_{2,2}  & \cdots  &  \W_{2,N} 
\\
\vdots  & \vdots   & \ddots  & \vdots
\\
\W_{N,1}  & \W_{N,2}  & \cdots  &   \W_{N,N} 
\end{bmatrix}
  \in \R^{N^2 \times N^2},
\end{eqnarray*}
where $\W_{i, k}   \in \R^{N \times N}$ whose elements are
\begin{equation*}
\left(  \W_{i, k}  \right)_{j, \ell} \ = \
\frac{ \Adj_{i, k} \M_{i, \ell}  }{ \sum_{v} \Adj_{v, k} \M_{v, \ell} } \mathcal{I}{( j = k )}.
\end{equation*}
The FSMT  matrix is modulated by the memory matrix $\M$.  When a memory-modulated random walker starts at the $i$th node and walks $L$ steps sequentially,  the expected number of times to visit the $j$th node is $\Pi^{(L, p, q)}_{i,j}$. When $p=q = 1$, the FSMT matrix is the FST proximity; that is, $\Pi^{(L, 1, 1)} = \Pi^{(L)}$.  This is the matrix implicitly used in node2vec~\cite{GroverL:16}.


\subsection{Warping function $g(\cdot)$}
\label{sec:warp}
The warping function warps the inner products of graph embeddings and normalizes the distribution of elements in $g(\F \Fw^T)$, providing an appropriate scale to link the proximity matrices in both node and embedding domains. The warping function is applied element-wise and monotonically increasing. 

We here consider the inverse Box-Cox transformation (IBC) as the warping function, defined as
\begin{equation*}
 g(x ) = {\rm IBC}( \gamma, x ) =
  \left\{ 
    \begin{array}{rl}
      ( 1 + \gamma x )^\frac{1}{\gamma}, & \gamma \neq 0;\\
      \exp ( x ), & \gamma = 0,
  \end{array} \right.
\end{equation*}
where $\gamma$ measures the nonlinearity. When $\gamma = 1$, $g(x) = 1+x$; this linear function is used in LapEigs~\cite{BelkinN:03}, diffusion maps~\cite{CoifmanL:06} and the exponential function  in SNE~\cite{HintonR:02}, DeepWalk~\cite{PerozziAS:14} and node2vec~\cite{GroverL:16}.


\subsection{Loss function $d(\cdot, \cdot)$}
\label{sec:loss}
The loss function quantifies the differences between the proximity matrices in two domains. Since $g(\cdot)$ is monotonic, when the rank of $g^{-1}(\Pi)$ is  $K$, we can find a pair of $\F, \Fw \in \R^{N \times K}$ satisfying  $\F \Fw^T = g^{-1}(\Pi)$ or $g(\F \Fw^T) = \Pi$. In this case, any valid loss function $d(\Pi, g(\F \Fw^T))$ should be zero, that is, the problem does not depend on the choice of loss function. When the rank of $g^{-1}(\Pi)$ is far-away from $K$, however, an appropriate loss function can still be beneficial to the overall performance.

\subsubsection{KL divergence}
We can use the KL-divergence based loss function to quantify the difference  between $\Pi$ and $g(\F \Fw^T)$.  To simplify the notation, let $\Y = g(\F \Fw^T) \in \R^{N \times N}$. We normalize the sum of each row in  both $\Pi$ and $\Y$ to be $1$ and use the KL divergence to measure the difference between two corresponding probability distributions.
The total loss between  $\Pi$ and $\Y$ is 
\begin{eqnarray*}
&&  d_{\rm KL}(\Pi, \Y) \ = \   - \sum_{i \in \V}  {\rm KL} ( \Pi_{i,:}  ||  \Y_{i,:} )
  \\
  & = &  - \sum_{i \in \V}  \sum_{j \in \V}  \left( \frac{ \Pi_{i,j} }{ \sum_{k} \Pi_{i,k}  }  \right) \log   \left(  {  \frac{ \Pi_{i,j} }{ \sum_{k} \Pi_{i,k}  }   }/{ \frac{ \Y_{i,j} }{ \sum_{k} \Y_{i, k} } }  \right),
\end{eqnarray*}
where $\Pi_{i,:}$ and $\Y_{i,:}$ denote the $i$th rows of $\Pi$ and $\Y$, respectively.  We then solve the following optimization problem:
 \begin{equation}
 \label{eq:emb}
	\F^{*}, \Fw^{*} 
	\ = \ \arg \min_{\F, \Fw \in \R^{N \times K}}   d_{\rm KL}(\Pi, \Y).
 \end{equation}
We use the gradient descent to optimize{\color{red}~\eqref{eq:emb}}. A stochastic algorithm can be used to accelerate the optimization, such as negative sampling, edge sampling or asynchronized stochastic gradient descent~\cite{MikolovCCD:13,MikolovSCCD:13,TangQWZYM:15}. 

\subsubsection{Warped Frobenius norm}
\label{sec:gem2}
The warped-Frobenius-norm based loss function is defined as
$$ d_{\rm wf}( \Pi, g(\F \Fw^T) = \left\| \F \Fw^T - g^{-1} (\Pi)  \right\|_F^2,
$$
where $g^{-1}(\cdot)$ is the inverse function of $g(\cdot)$. We thus solve the following optimization problem:
\begin{equation}
  \label{eq:emb2}
  \F^{*}, \Fw^{*} \ = \ \arg \min_{\F, \Fw \in \R^{N \times K}}  d_{\rm wf}^2( \Pi, g(\F \Fw^T).
 \end{equation}
The closed-form solution of{\color{red}~\eqref{eq:emb2}} can be efficiently obtained from the truncated singular value decomposition as follows:
\begin{subequations}
\begin{eqnarray*}
&& \bullet \Um_{[N \times K]}, \Sigma_{[K \times K]} \Vm_{[N \times K]} \leftarrow {\rm SVD}( g^{-1}\left( \Pi \right), K  );
\\
&& \bullet~{\rm Obtain}~ \F_{[N \times K]} \leftarrow \Um \Sigma^\frac{1}{2};
\\
&& \bullet~{\rm Obtain}~ \Fw_{[N \times K]} \leftarrow \Vm \Sigma^\frac{1}{2}.
\end{eqnarray*}
\end{subequations}



\subsection{Specific cases of GEM-D}
\label{sec:closed}
We now articulate the unifying power of GEM-D through special cases it encompasses.
\begin{myLem}
  \label{thm:dw}
  GEM-D includes DeepWalk as a special case with corresponding triple
  GEM-D$[ \Pi^{(L)},$ ${\rm exp}(x) , d_{\rm KL} (\cdot,\cdot)].$
  {\rm DeepWalk~\cite{PerozziAS:14} is a randomized implementation of
    GEM-D when the proximity matrix is the FST matrix, the warping
    function is the exponential function and the loss function is the
    KL divergence.}
\end{myLem}

\begin{myLem}
  GEM-D includes node2vec as a special case with the corresponding
  triple GEM-D$[ \Pi^{(L,p,q)},$ ${\rm exp}(x) , d_{\rm KL}
  (\cdot,\cdot)]$.
  \label{thm:node2vec} {\rm node2vec~\cite{GroverL:16} is a randomized
    implementation of GEM-D when the proximity matrix is the FSMT
    matrix, the warping function is the exponential function and the
    loss function is the KL divergence.}
\end{myLem}
In the original DeepWalk and node2vec, the authors did not specify the
closed-form proximity matrices. They estimated the proximity matrices
through random-walk simulations. Here we provide their closed forms,
which are equivalent to doing infinitely many random walk simulations.

\begin{myLem}
  GEM-D includes LINE as a special case with the corresponding triple
  GEM-D$[ \Pj$, $\frac{1}{1+{\rm exp}(-x) },$ $d_{\rm KL}
  (\cdot,\cdot)]$.
\label{thm:line}
{\rm LINE~\cite{TangQWZYM:15}   is an implementation of GEM-D when the proximity matrix is the transition matrix, the warping function is the sigmoid function and the loss function is the KL divergence.}
\end{myLem}

\begin{myLem} 
  GEM-D includes matrix factorization as a special case with the
  corresponding triple GEM-D$[ \Adj,$ $x, d_{\rm wf}
  (\cdot,\cdot)]$.
\label{thm:mf}
{\rm Matrix factorization~\cite{KorenBV:09} is an implementation of GEM-D when the proximity matrix is the adjacency matrix, the warping function is the linear function and the loss function is the warped Frobenius norm.}
\end{myLem}

\section{Algorithm: UltimateWalk}
Based on the  GEM-D framework and our experiments,  we observe that (1) the transition matrix with high-order proximity matters, while the memory factor does not (at least not much); (2) the warping function is the key contributor to the performance and the exponential function usually performs the best; and (3) once the node proximity function and warping function are properly selected, the distance function does not matter. We thus choose a triple: the FST matrix, exponential function and warped Frobenius norm to propose \fbox{UltimateWalk = GEM-D$[ \Pi^{(L)},   {\rm exp}(x) , d_{\rm wf} (\cdot,\cdot)]$}.

\subsection{Closed form}
The UltimateWalk algorithm solves
\begin{equation*}
	\F^{*}, \Fw^{*} \ = \ \arg \min_{\F, \Fw \in \R^{N \times K}}  \left\| \F \Fw^T  - \log \left( \Pi^{(L)} \right) \right\|_F^2,
\end{equation*}
where $L$ is the diameter of the given graph (based on our experiments, we set our default value to be $7$). To ensure numerical correctness, during the computation, we take the log of non-zero elements and replace others with $-c$ (the default value for $c$ is $100$), that is, we set 
\begin{equation}
  \label{eq:Z}
  \Z \ = \ \log ( \Pi^{(L)} ) \mathcal{I}{( \Pi^{(L)} \neq 0 )} -c \mathcal{I}{( \Pi^{(L)} =0 )} \in \R^{N \times N},
\end{equation}
 where $\mathcal{I}{(\cdot)}$ is the element-wise indicator function. The solution is obtained from the truncated singular value decomposition of $\Z$. UltimateWalk has the closed-form solution; the choice of each building block in UltimateWalk is based on an insightful understanding of GEM-D. We find that $\log ( \Pi^{(7)} )$ maximizes the empirical performance by normalizing the distribution of elements in $\log ( \Pi^{(7)} )$. Specifically, $L =7 $ is usually the diameter of a social network and $\Pi^{(7)}$ properly diffuses information to all nodes without introducing echo and redundancy. Empirically, we observe that the elements in $\Pi^{(7)}$ approximately follow a log-normal distribution and SVD can only preserve high-magnitude elements. The logarithm factor aids in normalizing the distribution, such that SVD preserves more information in $\Pi^{(7)}$ (see Figure{\color{red}~\ref{fig:hist})}.

\subsection{Scalable implementation}
\label{sec:LS}
Since $\Pi^{(L)} $ is not necessarily a sparse matrix, when the graph is large, we cannot afford to compute the closed form of $\log ( \Pi^{(L)} )$, which if of order $O(N^2)$; instead, we estimate it via random-walk simulation. Specifically,  to estimate the $(i,j)$th element $\Pi^{(L)}_{i,j}$,  we run $L$-step random walks from the $i$th node for $m$ independent trials and record the number of times that the $j$th node has been visited, denoted by $\Ss_{i,j}$. In each trial, we start from the $i$th node and walk $L$ steps randomly and sequentially. The estimate of $\Pi^{(L)}_{i,j}$ is then $\widetilde{\Pi}^{(L)}_{i,j} = \Ss_{i,j} /m$.  To work with a sparse matrix, we further obtain a proxy matrix $\widetilde{\Z} \in \R^{N \times N}$ defined as
$$\widetilde{\Z} = \left( \log ( \widetilde{\Pi}^{(L)} ) + c \right) \mathcal{I}{( \widetilde{\Pi}^{(L)} \neq 0 )} \in \R^{N \times N}.$$
When $m \rightarrow +\infty$,  $\widetilde{\Pi}^{(L)}_{i,j} \rightarrow \Pi^{(L)}_{i,j}$ and $\widetilde{\Z} \rightarrow \Z + c \one_N \one_N^T$, where $\Z$ follows from{\color{red}~\eqref{eq:Z}}. We obtain the solution by the singular value decomposition of the proxy matrix $\widetilde{\Z}$. Since $\widetilde{\Pi}^{(L)}$ is sparse, $\widetilde{\Z}$  is sparse and  the computation of singular value decomposition is cheap; we call this approach~\emph{scalable UltimateWalk}.  We add constant $c$ to all the elements in $\log (\widetilde{\Pi}^{(L)} )$, which will introduce a constant singular vector and does not influence the shape of the graph embedding.

\begin{myLem} UltimateWalk is scalable.
\label{thm:complexity}
{\rm The time and space complexities of scalable UltimateWalk are at most $O(K \max( |\E|,$  $NLm) )$ and $O(\max( |\E|,$ $NLm))$, respectively, where $N$ is the size of graph, $K$ is the dimension of graph embedding, $L$ is the walk length and $m$ is the number of random-walk trials.}
\end{myLem}
\begin{proof}
Each random walk from a given node can visit at most $L$ unique nodes and $m$ such walks would give non-zero proximity for a node to at most $Lm$ neighbors. Thus, the space requirement is $O(\max(|\E|, NLm))$. The complexity of the truncated singular value decomposition of $\widetilde{\Z}$ is then $O(K \max( |\E|,$ $ NLm) )$.
\end{proof}
In practice, a few nodes are visited more frequently than others; thus, the sparsity of the proxy matrix $\widetilde{\Z}$ is much smaller than $O(NLm)$.

\subsection{Speeding~up~via~data~splitting}
Instead of using all the random walks to estimate $\Pi^{(L)}$ and then taking the logarithm operator, we run $T$ threads in parallel, estimate $\log ( \Pi^{(L)} )$ in each thread, and finally average the results in all threads to obtain a more accurate estimation of $\log ( \Pi^{(L)} )$. The advantages are two-fold: (1) we parallelize the procedure; (2) we focus on estimating $\log ( \Pi^{(L)} )$, not $\Pi^{(L)}$. As mentioned earlier, the empirical observations shows that each element in $\log ( \Pi^{(L)} )$ approximately follows the normal distribution. The following lemma ensures that the data splitting procedure provides the maximum likelihood estimator of $\log ( \Pi^{(L)} )$.



\begin{myLem} Average provides the maximum likelihood estimator.
  \label{thm:boot} {\rm Denote the $(i,j)$th element of
    $\widetilde{\Pi}^{(L)}$ in the $t$th thread as $x_t$. Let $\log
    (x_1)$, $\log (x_2),$ $\cdots$, $\log (x_T)$ be independent and
    identically distributed random variables, which follow the normal
    distribution with mean $\log({\Pi_{i,j}})$ and variance
    $\sigma^2$. Then, the maximum likelihood estimator of
    $\log({\Pi_{i,j}})$ is $ \sum_{t=1}^T \log (x_t) / T.  $}
\end{myLem}
The proof follows from the maximum likelihood estimator of the log-normal distribution. Guided by Lemma{\color{red}~\ref{thm:boot}} and considering numerical issues, we obtain the proxy matrix $\widetilde{ \Z}_{t}$ in the $t$th thread, average the proxy matrices in all the threads and obtain
$
\widetilde{\Z}_{\rm avg} =  \sum_{t=1}^T \widetilde{\Z}_{t}/T,
$
 which estimates $\log ( \Pi^{(L)} )+c$.
Finally, we do the singular value decomposition of $\widetilde{\Z}_{\rm avg}$.

\section{Experimental Results}
\label{sec:exp}
The main goal here is to empirically demystify how each building block of GEM-D influences the overall performance. We plugin various proximity matrices and warping functions to GEM-D and test their performance. To avoid the randomness and suboptimal solutions, we adopt the warped Frobenius norm and its corresponding SVD solution, which is unique and deterministic.

We focus on the task of node classification; we observe a graph where each node has one or more labels, we randomly sample the labels of a certain fraction (labeling ratio) of nodes  and we aim to predict the labels of the remaining nodes. We repeat this process for $20$ times and report the average results in terms of  Macro  and Micro F1 scores~\cite{TangL:09a}. For all the graph embeddings, we 
choose dimension size $K = 64$ and implement the classification by using a one-vs-rest logistic regression of LibLinear~\cite{FanCHWL:08}.

\subsection{Datasets}

We use the following datasets:
\begin{itemize}
\item {\tt Kaggle}
  1968\footnote{~\url{https://www.kaggle.com/c/learning-social-circles/data}}.
  This is a social network of Facebook users. The labels are the
  social circles among the users. The network has $277$ nodes, $2,321$
  edges and $14$ labels.

\item {\tt
    BlogCatalog}\footnote{~\url{http://socialcomputing.asu.edu/datasets/BlogCatalog3}}~\cite{TangL:09a}.
  This is a network of social relationships provided by bloggers on
  the BlogCatalog website. The labels are the topic categories
  provided by the bloggers. The network has $10,312$ nodes, $333,983$
  edges and $39$ labels.

\item {\tt
    Flickr}\footnote{~\url{http://socialcomputing.asu.edu/datasets/Flickr}}~\cite{TangL:09b}.
  This is a social network of Flickr users who share photos.  The
  labels are the interest groups of users. The network has $80,513$
  nodes, $5,899,882$ edges and $195$ labels.

\item  {\tt  U.S. Patent 1975-1999}\footnote{~\url{http://snap.stanford.edu/data/cit-Patents.html}}~\cite{LeskovecKF:05}. U.S. patent dataset is maintained by the National Bureau of Economic Research. This is a citation network among the patents granted between 1975 and 1999. The network has $3,774,768$ nodes, $16,518,948$ edges.
\end{itemize}


\subsection{Which factor matters most?}
We aim to demystify a graph embedding, that is,   investigate its building blocks from four aspects: randomized implementation, walk length, nonlinearity of the warping function and memory factors.

\subsubsection{Random walk vs. closed-form}
Does the power of DeepWalk  and node2vec come from the flexible exploration via random walks?  In Section{\color{red}~\ref{sec:closed}}, we have shown that DeepWalk  and node2vec are randomized implementations of GEM-D as the corresponding proximity matrix is approximated via random walks. We compare the original DeepWalk and node2vec with a corresponding closed-form solution.

We compare the classification performance of DeepWalk, node2vec, the closed form of UltimateWalk and scalable UltimateWalk on three datasets: Kaggle 1968, BlogCatalog and Flickr.  For DeepWalk, we set the walk length to be $7$;
for node2vec, we set the walk length to be $7$ and the memory factors $p = 0.25 ,q = 0.25$, as recommended in~\cite{GroverL:16}.  We set the labeling ratio to $50\%$; that is, a half of the nodes is used for training and the rest for testing. The classification performance is shown in Figure{\color{red}~\ref{fig:blog_converge}}. In each plot, the $x$-axis is the running time and the $y$-axis is the F1 score (either Macro or Micro). The blue, yellow and red curves are DeepWalk (dw), node2vec  (n2v) and scalable UltimateWalk (s-uw), which all rise as the running time grows; black diamonds denotes the closed-form of UltimateWalk~(uw).

 We see that (1) as the closed-form solution, UltimateWalk provides the most accurate classification performance; with increasing running time, DeepWalk, node2vec and scalable UltimateWalk  converge to UltimateWalk, which validates Lemma{\color{red}~\ref{thm:dw}}; (2) with increasing size of the input graph, the closed-form UltimateWalk takes much longer in terms of running time; (3) at a similar level of classification performance, the proposed scalable UltimateWalk runs much faster than DeepWalk and node2vec. 
 
\mypar{Summary} Randomness does not improve accuracy, but makes the algorithm scalable.

\begin{figure}[htb]
  \begin{center}
    \begin{tabular}{cc}
        \includegraphics[width=0.3\columnwidth]{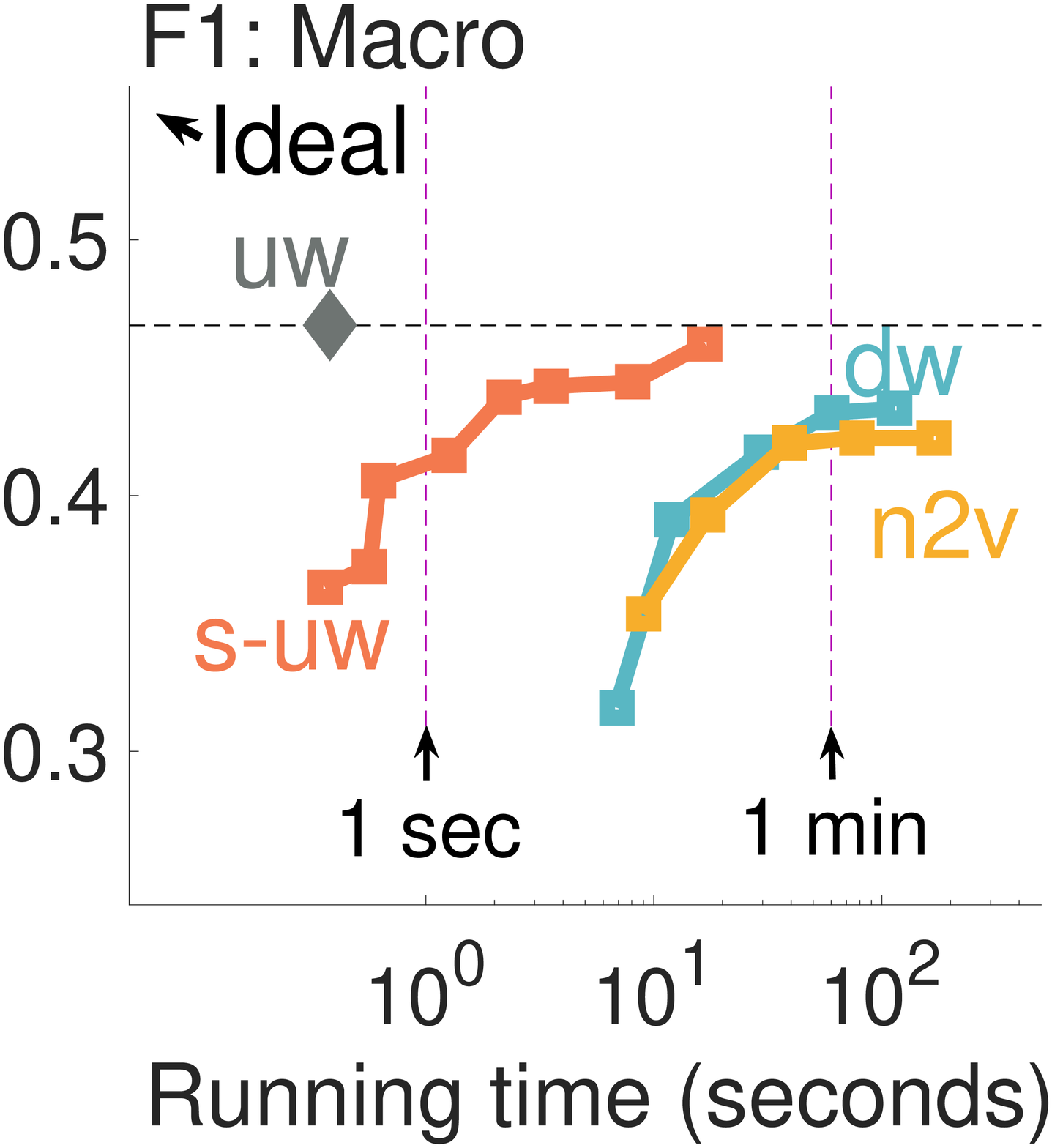}   & \includegraphics[width=0.3\columnwidth]{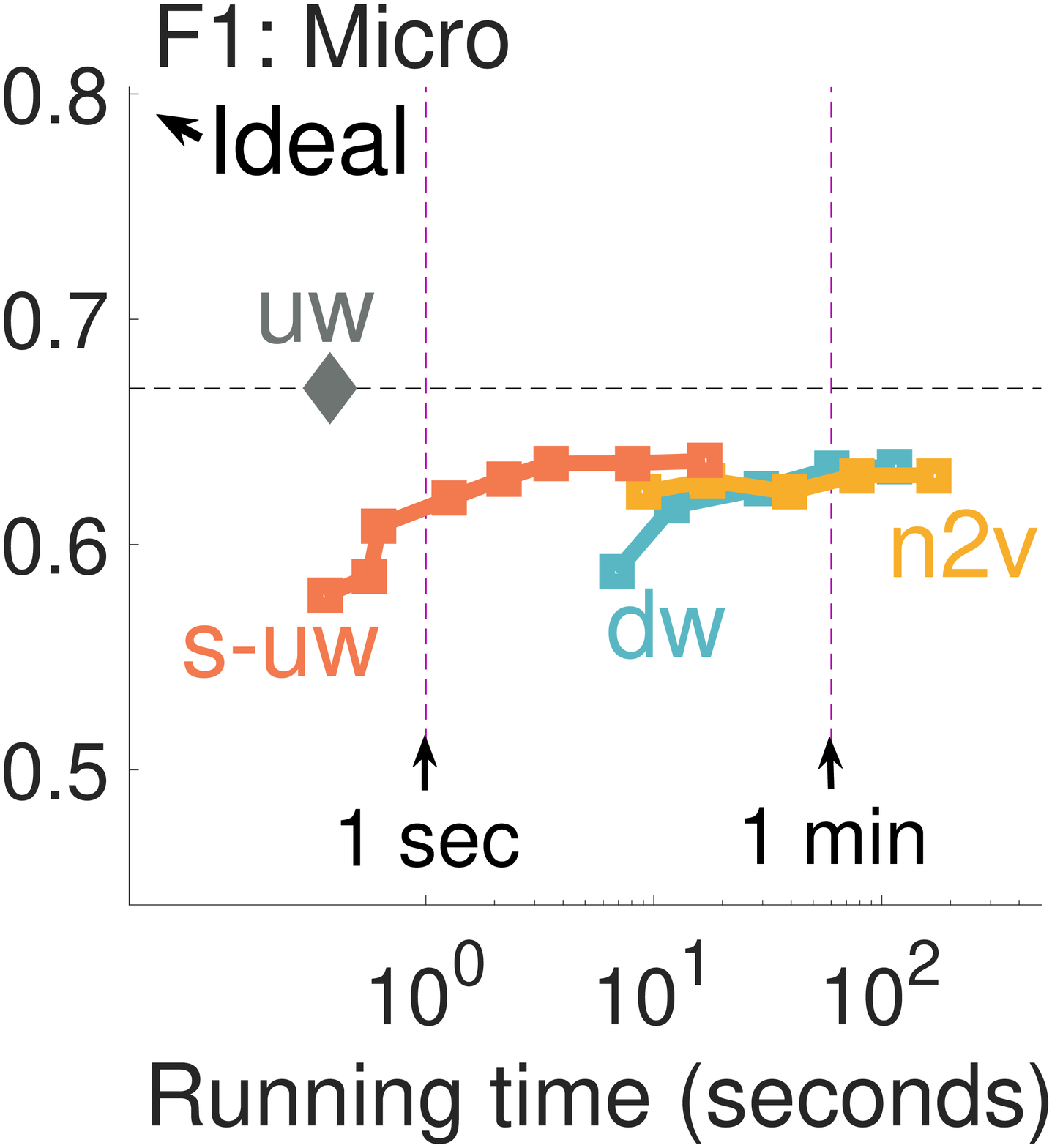}
\\
 {\small (a) Macro in Kaggle 1968. }  &   {\small (b) Micro  in Kaggle 1968. }
 \\
\\
\includegraphics[width=0.3\columnwidth]{figures/comparison/blog_macro_50.eps}   & \includegraphics[width=0.3\columnwidth]{figures/comparison/blog_micro_50.eps}
\\
 {\small (c) Macro in BlogCatalog. }  &   {\small (d) Micro  in BlogCatalog. }
\\
\\
\includegraphics[width=0.3\columnwidth]{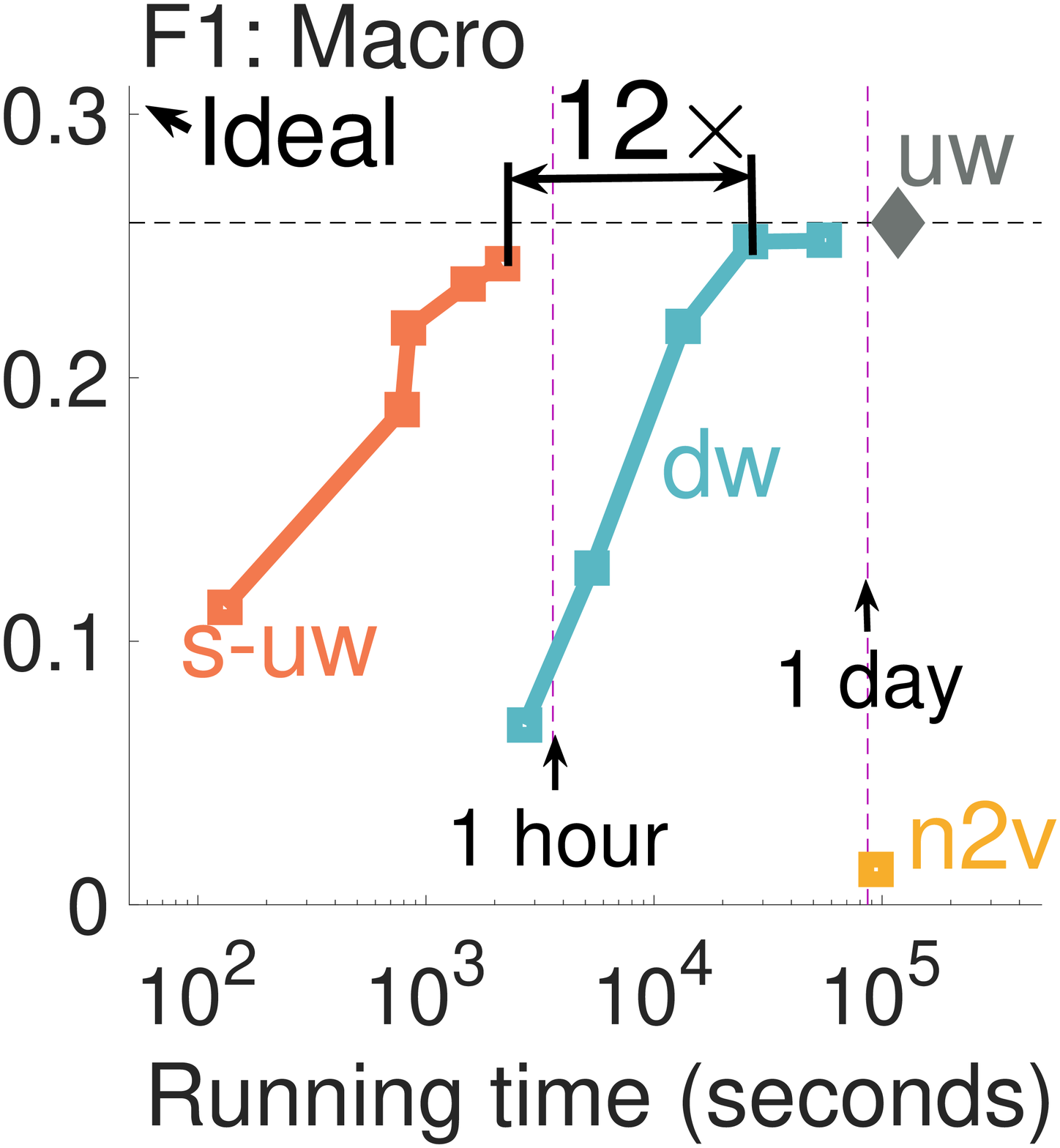}   & \includegraphics[width=0.3\columnwidth]{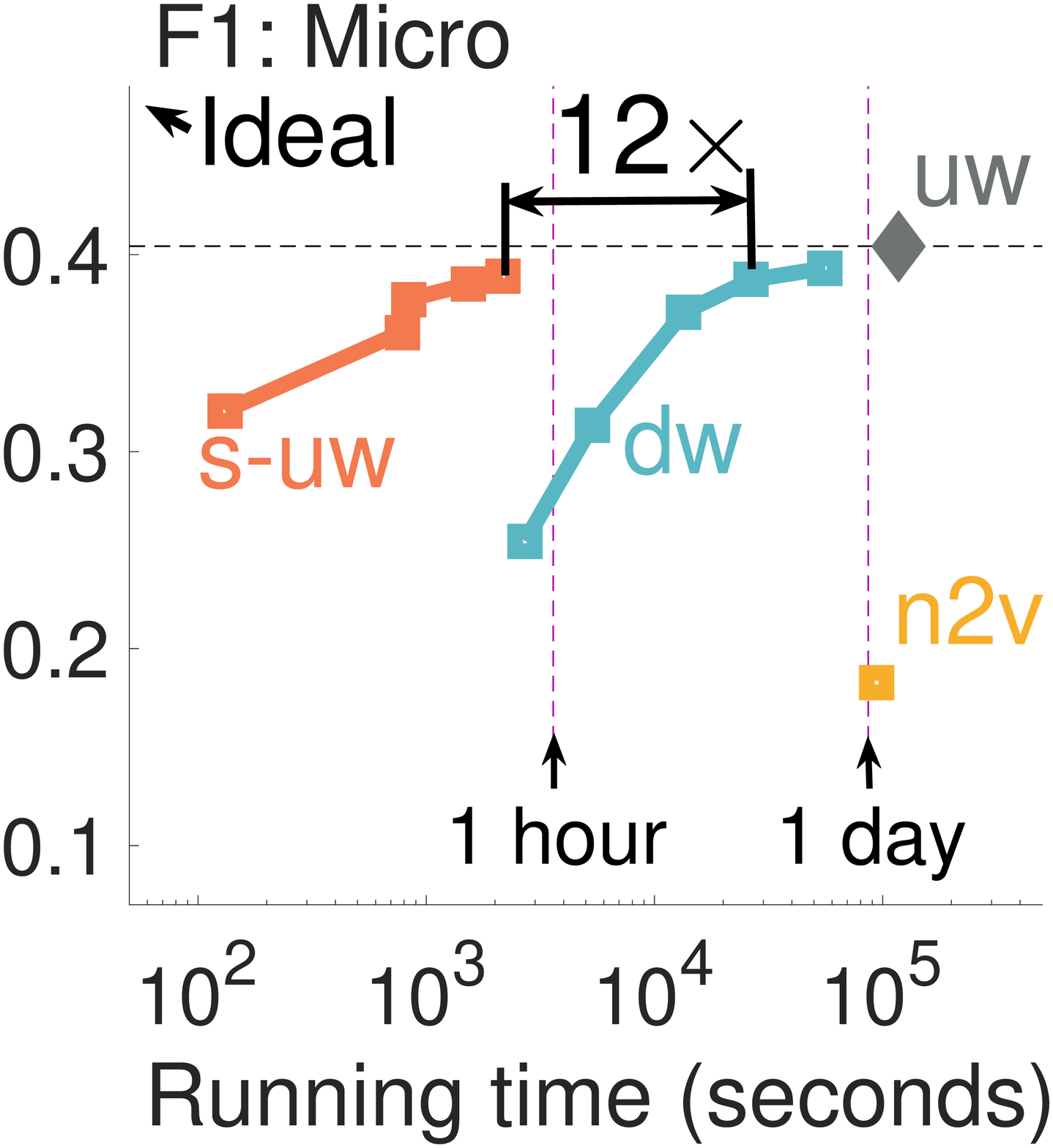}
\\
 {\small (e) Macro in Flickr. }  &   {\small (f) Micro  in Flickr. }
 \\
    \end{tabular}
  \end{center}
  \caption{\label{fig:blog_converge} \emph{UltimateWalk outperforms competition.} The closed-form UltimateWalk (uw, grey diamond) is the most accurate; scalable UltimateWalk (s-uw, red) is much faster than DeepWalk (dw, blue) and node2vec (n2v, yellow). }
\end{figure}

\begin{figure}[htb]
  \begin{center}
    \begin{tabular}{cc}    
     \includegraphics[width=0.3\columnwidth]{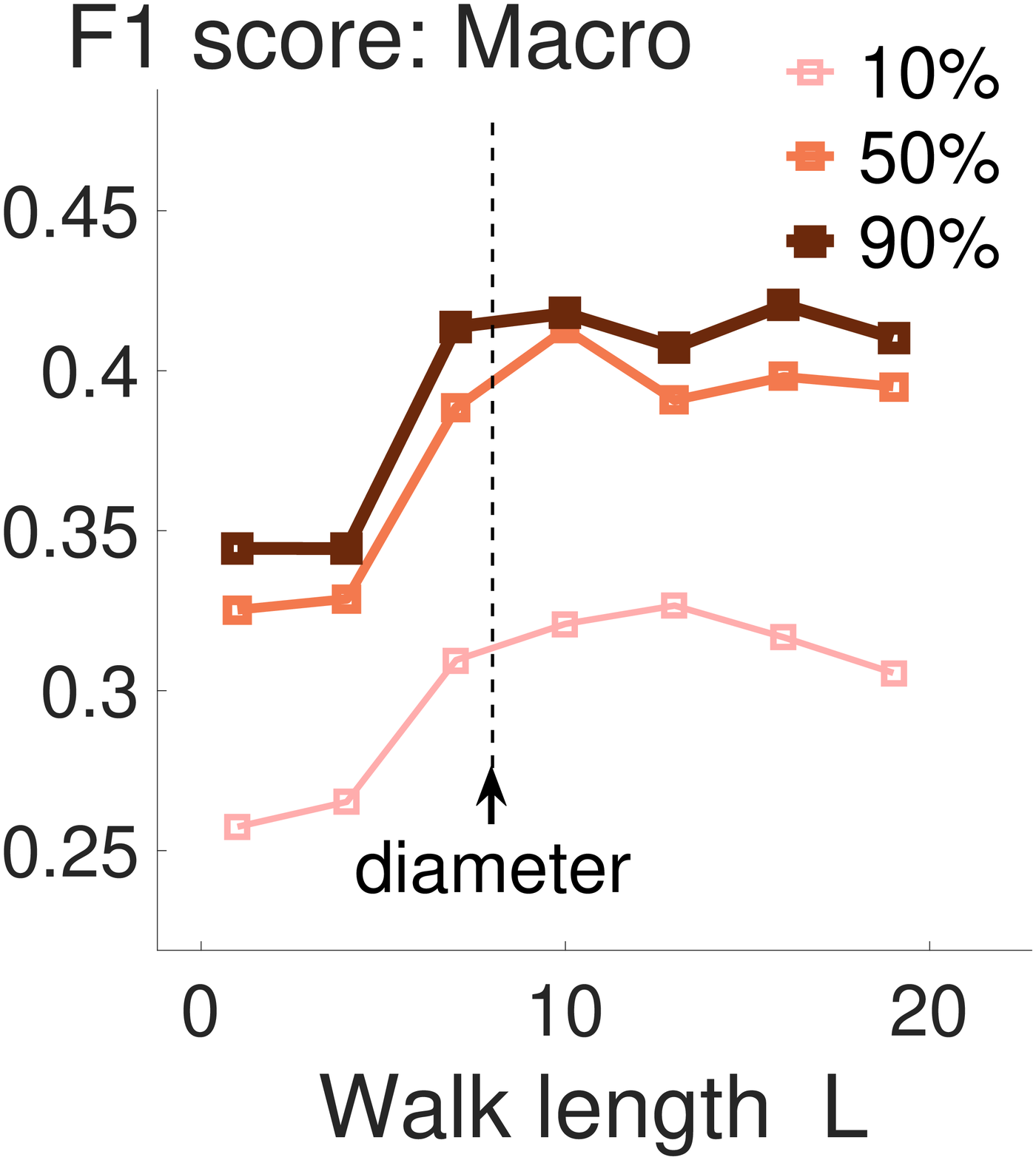}   & \includegraphics[width=0.3\columnwidth]{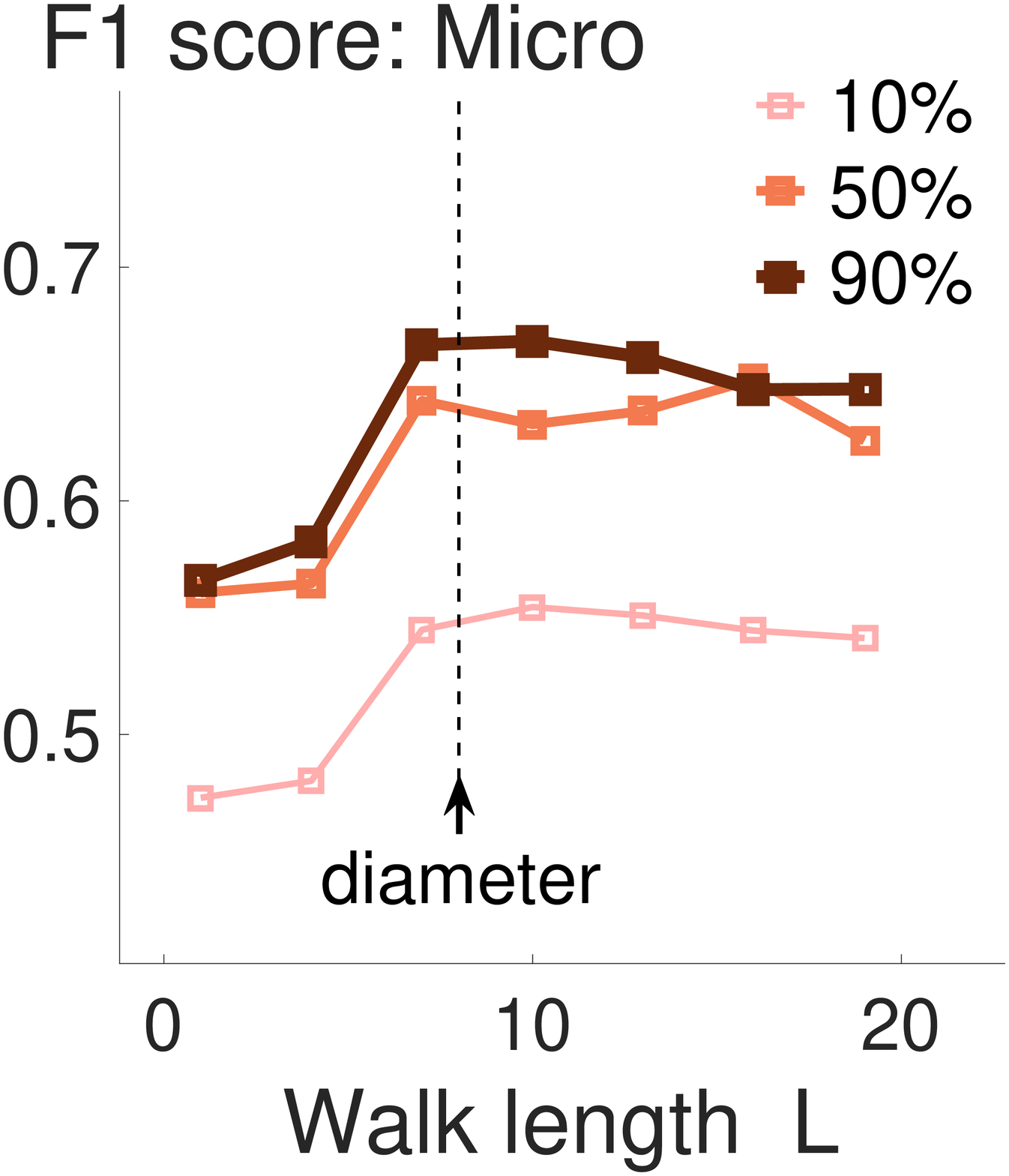} 
\\
 {\small (a) Macro; Kaggle 1968. }  &   {\small (b) Micro; Kaggle 1968. }
 \\
 \\
 \includegraphics[width=0.3\columnwidth]{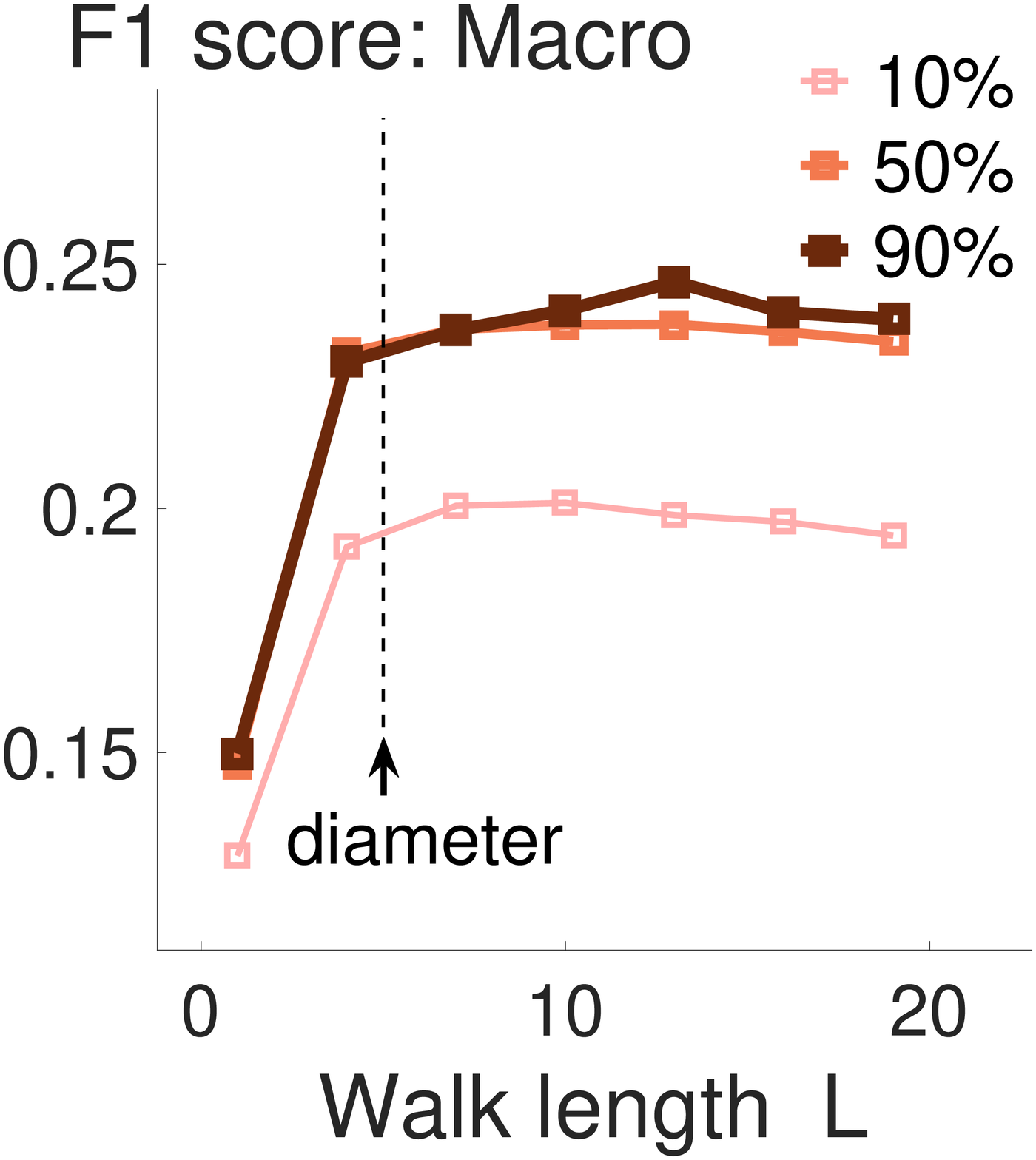}   & \includegraphics[width=0.3\columnwidth]{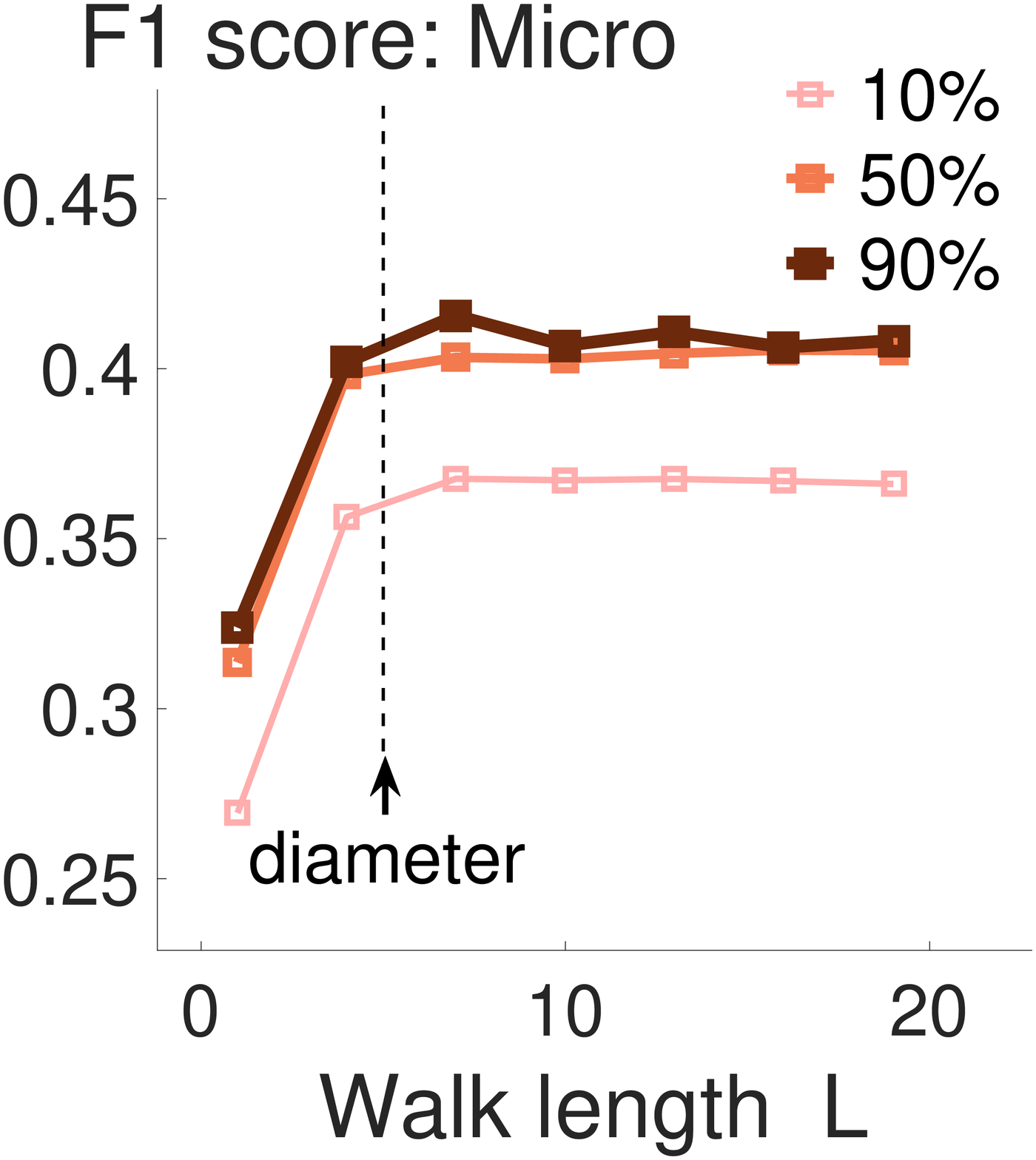} 
\\
 {\small (c) Macro; BlogCatalog. }  &   {\small (d) Micro; BlogCatalog. }
    \end{tabular}
  \end{center}
  \caption{\label{fig:blog_step} ~\emph{The walk length matters.} The walk length influences the classification performance. An empirical choice of walk length is the diameter of input graph (Kaggle 1968 and BlogCatalog). }
\end{figure}

\subsubsection{One step vs. Long walk}
In the FST matrix, a tuning parameter is the walk length, which controls the range of random walks.  We consider the closed-form UltimateWalk with the FST matrix $\Pi^{(L)}$, varying $L$ from $1$ to $20$. The goal is to understand how the walk length influences the overall performance.

We compare the classification performance on Kaggle 1968 and Blogcatalog in Figure{\color{red}~\ref{fig:blog_step}}, where the $x$-axis is the walk length and the $y$-axis is the F1 score. In each plot, light, medium and dark red curves show the performance under the labeling ratios of $10\%$, $50\%$ and $90\%$, respectively. 

We see that the walk length has a significant impact on the classification performance; it increases considerably when the walk length increases in the beginning and then drops slowly when the number of walks increases even further. The intuition is as follows: when the  walk length is too small, only low-order proximity information is obtained and the community-wise information is hard to capture; when the walk length is too large, high-order proximity information is over-exposed; that is, the proximity matrix tends to be an all-one matrix and does not contain any structure information. Empirically, we find that the sweet spot of
the walk length is approximately the diameter of a graph, which is usually around $7$ for a social network. Note that the diameters of Kaggle 1968 and BlogCatalog are $8$ and $5$, respectively.

\mypar{Summary} The walk length matters; the choice should be the diameter of the graph.

\begin{figure}[htb]
  \begin{center}
    \begin{tabular}{cc}
\includegraphics[width=0.3\columnwidth]{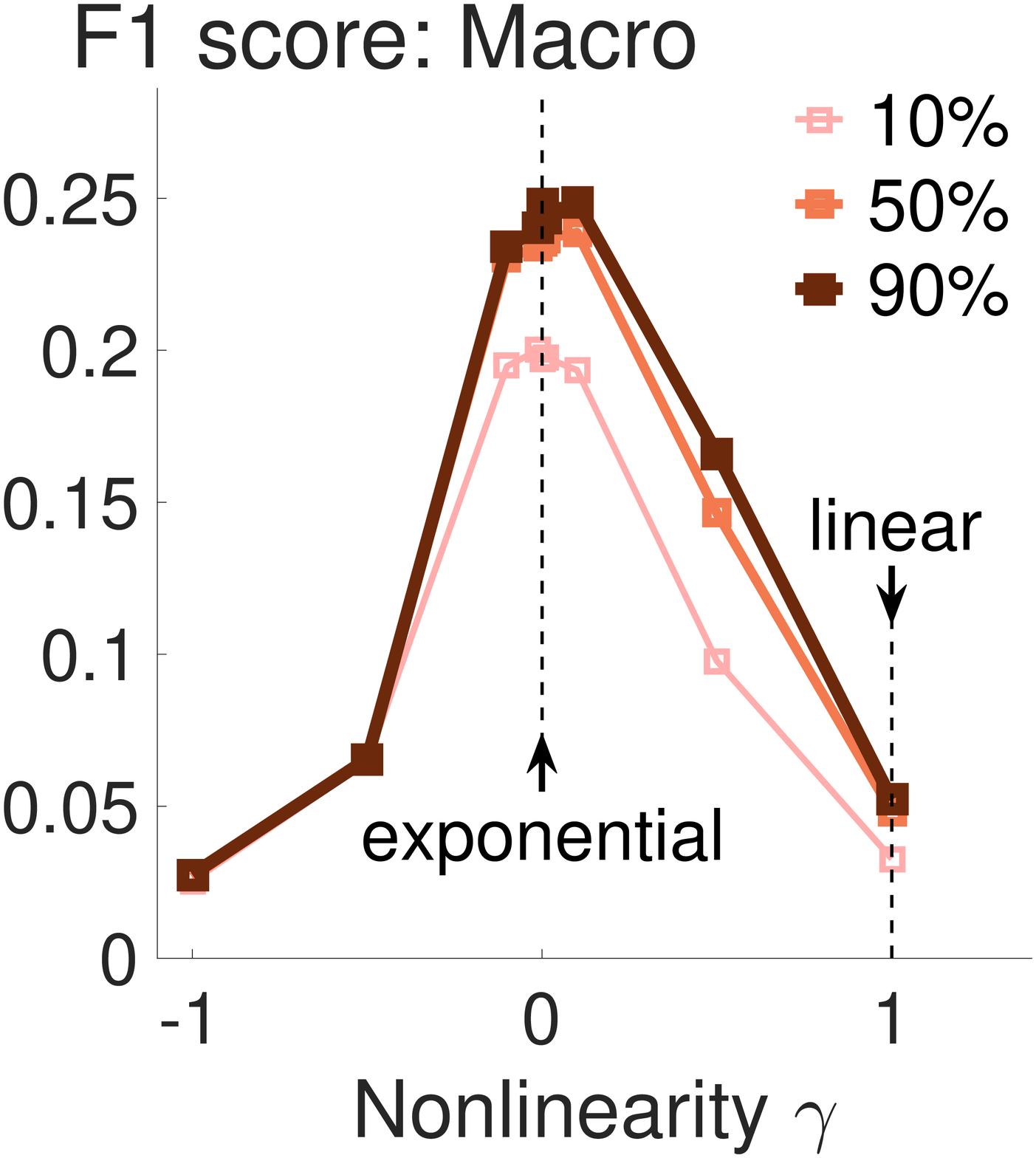}   & \includegraphics[width=0.3\columnwidth]{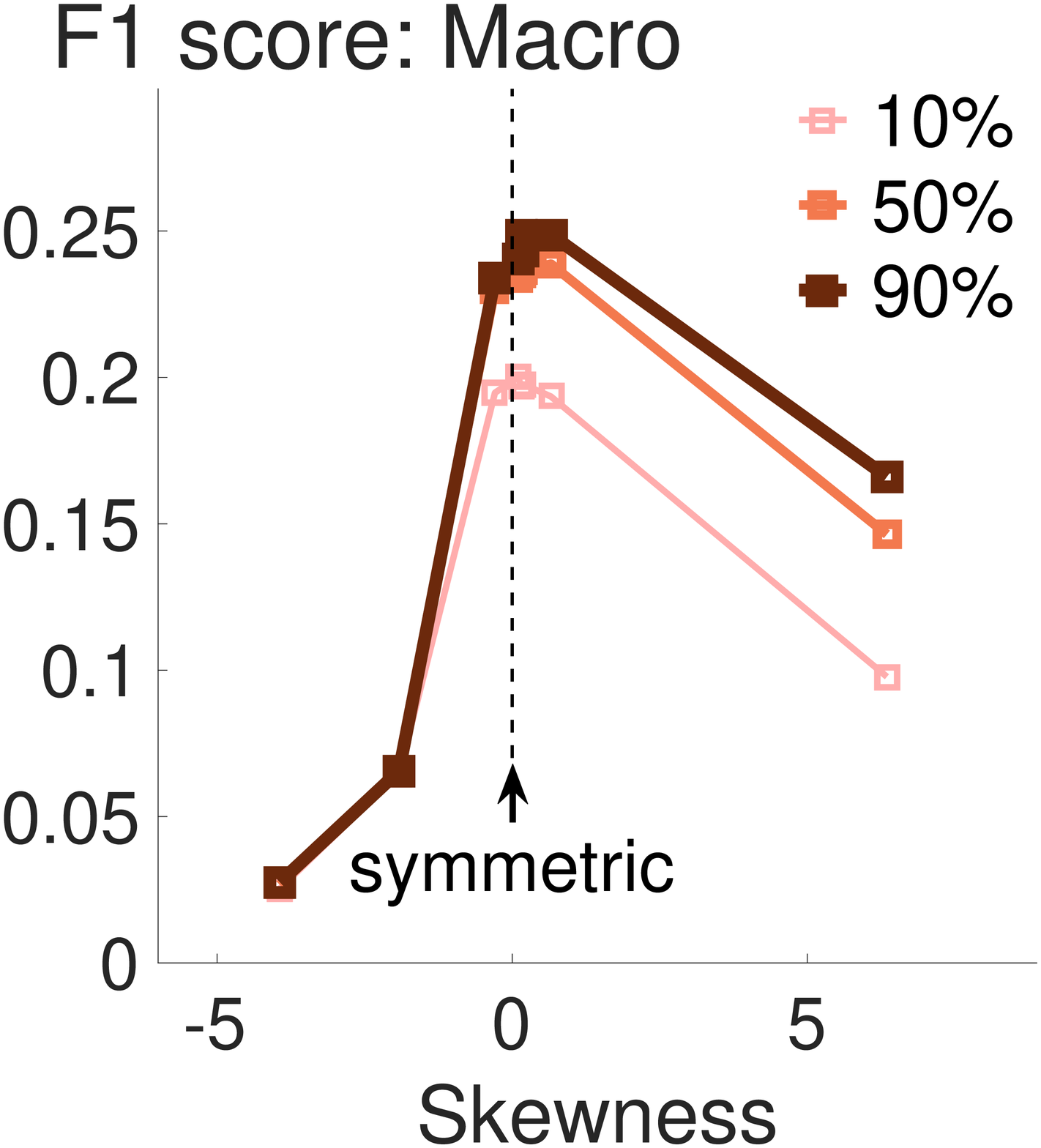} 
\\
 {\small (a) Nonlinearity; Macro. }  &   {\small (b) Skewness; Macro. }
 \\
  \\
 \includegraphics[width=0.3\columnwidth]{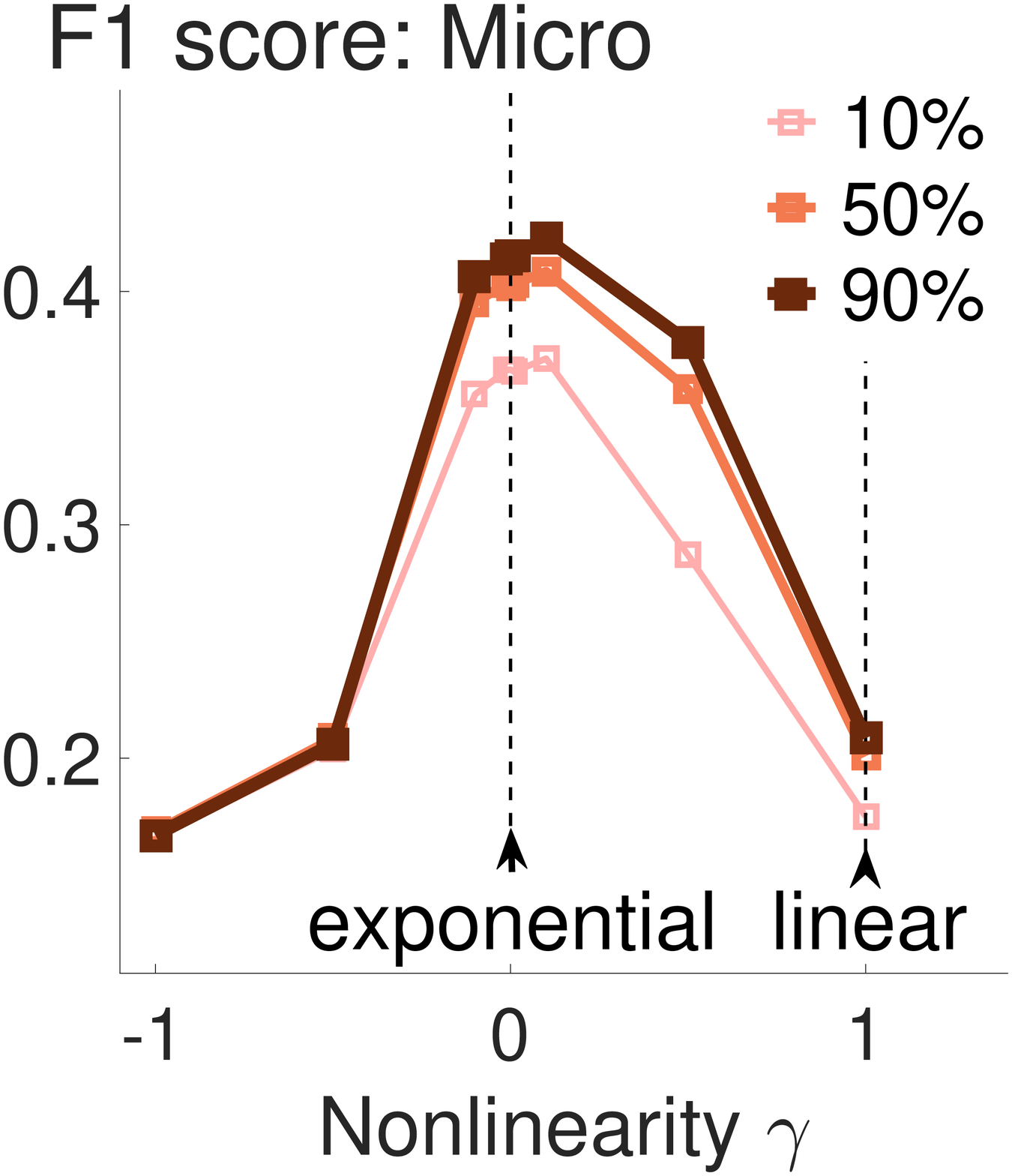}   & \includegraphics[width=0.3\columnwidth]{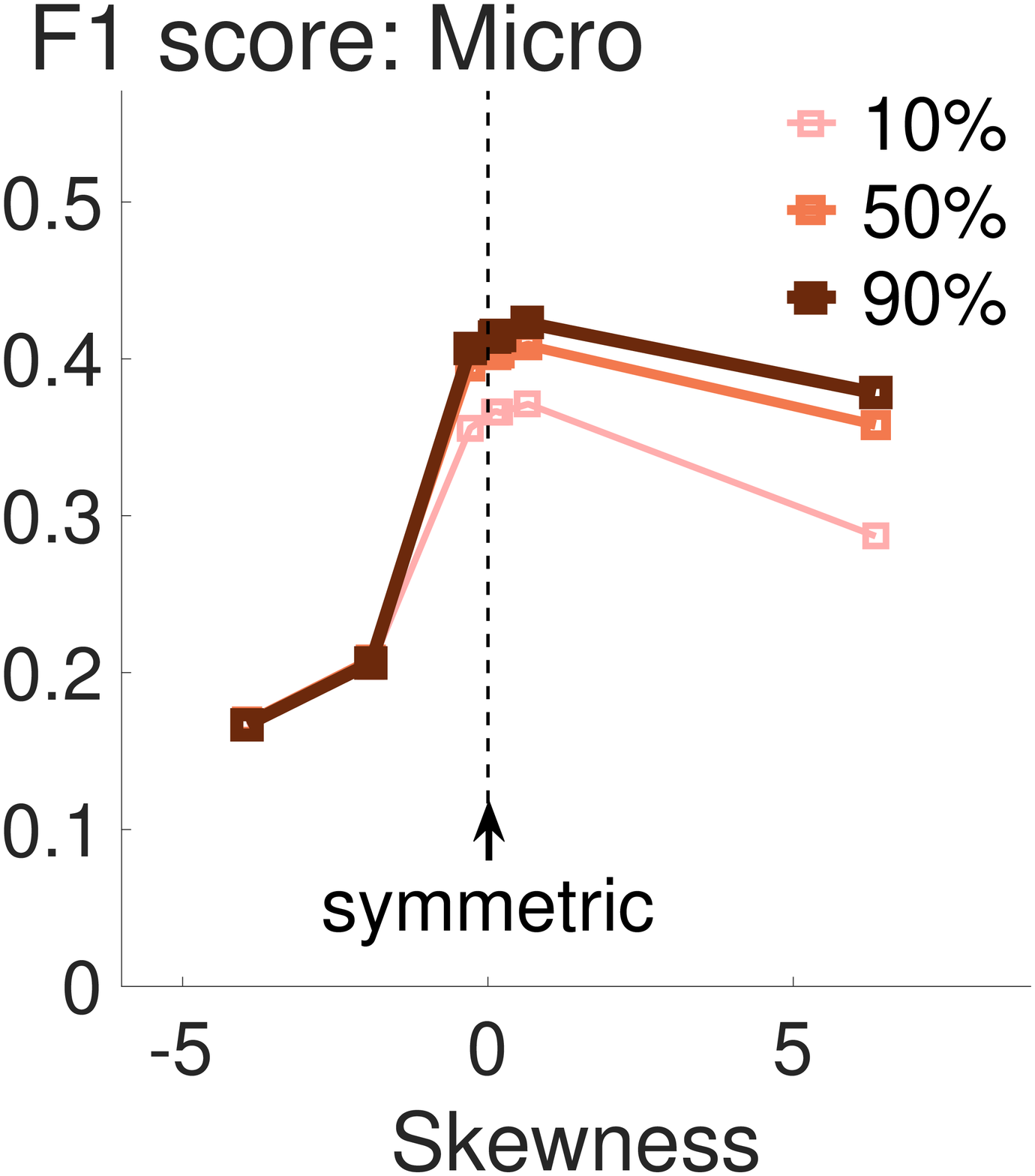} 
\\
 {\small (c) Nonlinearity; Micro. }  &   {\small (d) Skewness; Micro. }
    \end{tabular}
  \end{center}
  \caption{\label{fig:blog_nonlinear} \emph{Nonlinearity matters.}  Nonlinearity and skewness significantly influence the F1 score (Blogcatalog). The warping function normalizes  the distribution of all elements in the proximity matrix; when the distribution is symmetric, skewness is zero and the best performance is achieved. }
\end{figure}

\subsubsection{Linear vs. Nonlinear}
Many previous embedding methods, such as SNE, DeepWalk and node2vec, adopt the softmax model, which corresponds to an exponential warping function $g(x) = \exp(x)$. How does it differ from a linear warping function $g(x) = x$? We consider GEM-D with the FST matrix $\Pi^{(L)}$ with $L = 7$ and vary the nonlinear warping function  $g(x) = (1+\gamma x)^\frac{1}{\gamma}$ from $-1$ to $1$; note that $g(x) = \exp(x)$ when $\gamma = 0$.

We compare the classification performances on Blogcatalog in Figure{\color{red}~\ref{fig:blog_nonlinear}} (a) and (c), where the $x$-axis is the nonlinearity parameter $\gamma$ 
and the $y$-axis is the F1 score. In each plot, light, medium and dark red  curves show the performance under the labeling ratios of $10\%$, $50\%$ and $90\%$, respectively.  

We see that  the nonlinearity has a significant impact on the classification performance. A linear function ($\gamma = 1$) results in poor performance. The classification accuracy  increases significantlyl as more nonlinearity is introduced ($\gamma$ decreases). The exponential function ($\gamma = 0$) almost always provides the best performance. After that, however, adding more nonlinearity deteriorates the classification performance. These results validate the superiority of the (nonlinear) exponential model over the linear model.



In UltimateWalk, we consider factoring $\log (\Pi)$. The logarithm  is introduced to rescale the distribution of $\Pi$ and amplify small structural discrepancies. On the other hand, when we introduce more  nonlinearity than the logarithm, all the values tend to be the same. How should we then choose the nonlinearity parameter $\gamma$? Empirically, we find that the skewness of the element distribution of $g^{-1} (\Pi)$ influences the classification performance. Let $x$ be a random variable. We define 
$
 {\rm skewness} \ = \ {\mathbb{E} (x-\mu)^3}/{\sigma^3},
$
where $\mu$ and $\sigma$ are the mean and standard deviation of $x$. Figure{\color{red}~\ref{fig:blog_nonlinear}} (b) and (d) show the relations between skewness and  classification performance; we see that the classification performance reaches its highest value when the skewness is zero. To illustrate the intuition, Figure{\color{red}~\ref{fig:hist}} shows the histogram of elements in the proximity matrix $g^{-1} (\Pi^{(7)})$ under various $\gamma$, where $x$-axis is the value and $y$-axis counts the number of elements associated with the value; when $\gamma = 0.5$, slight nonlinearity is introduced and, as shown in Figure{\color{red}~\ref{fig:hist}} (c), the distribution is highly left-skewed, meaning most pairwise relationships are weak. This is common in social networks due to the power-law phenomenon~\cite{FaloutsosFF:99}. The logarithm factor rescales the distribution to make it symmetric. An empirical choice of the optimal nonlinearity could be to randomly sample a subset of elements in $\Pi$ and choose the $\gamma$ that makes the distribution symmetric.

\mypar{Summary} The warping function matters; the choice should be the one that makes the distribution of elements in the proximity matrix symmetric.

\begin{figure}[htb]
  \begin{center}
    \begin{tabular}{ccc}
  \includegraphics[width=0.2\columnwidth]{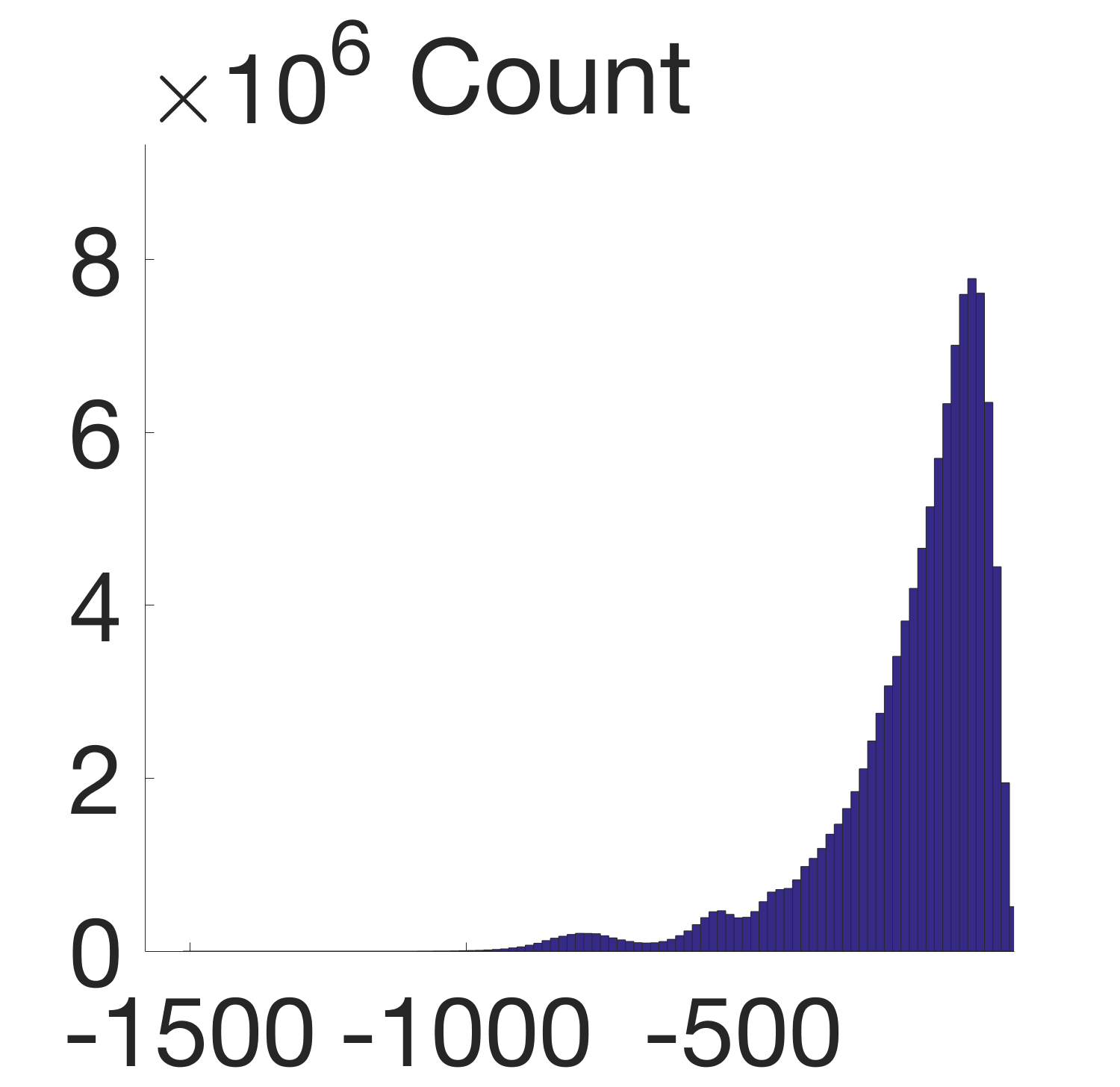}  &
  \includegraphics[width=0.2\columnwidth]{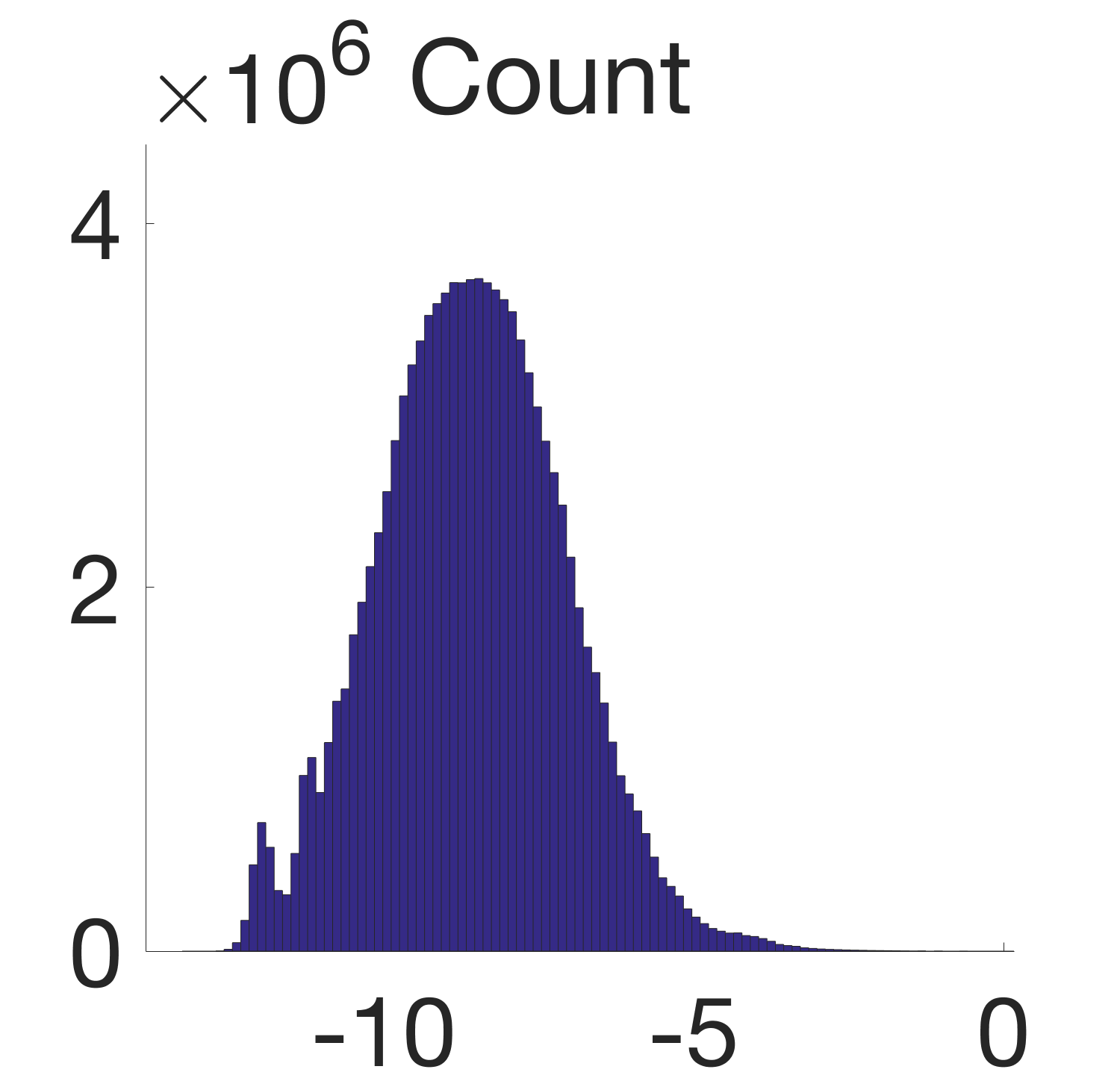}  &
  \includegraphics[width=0.2\columnwidth]{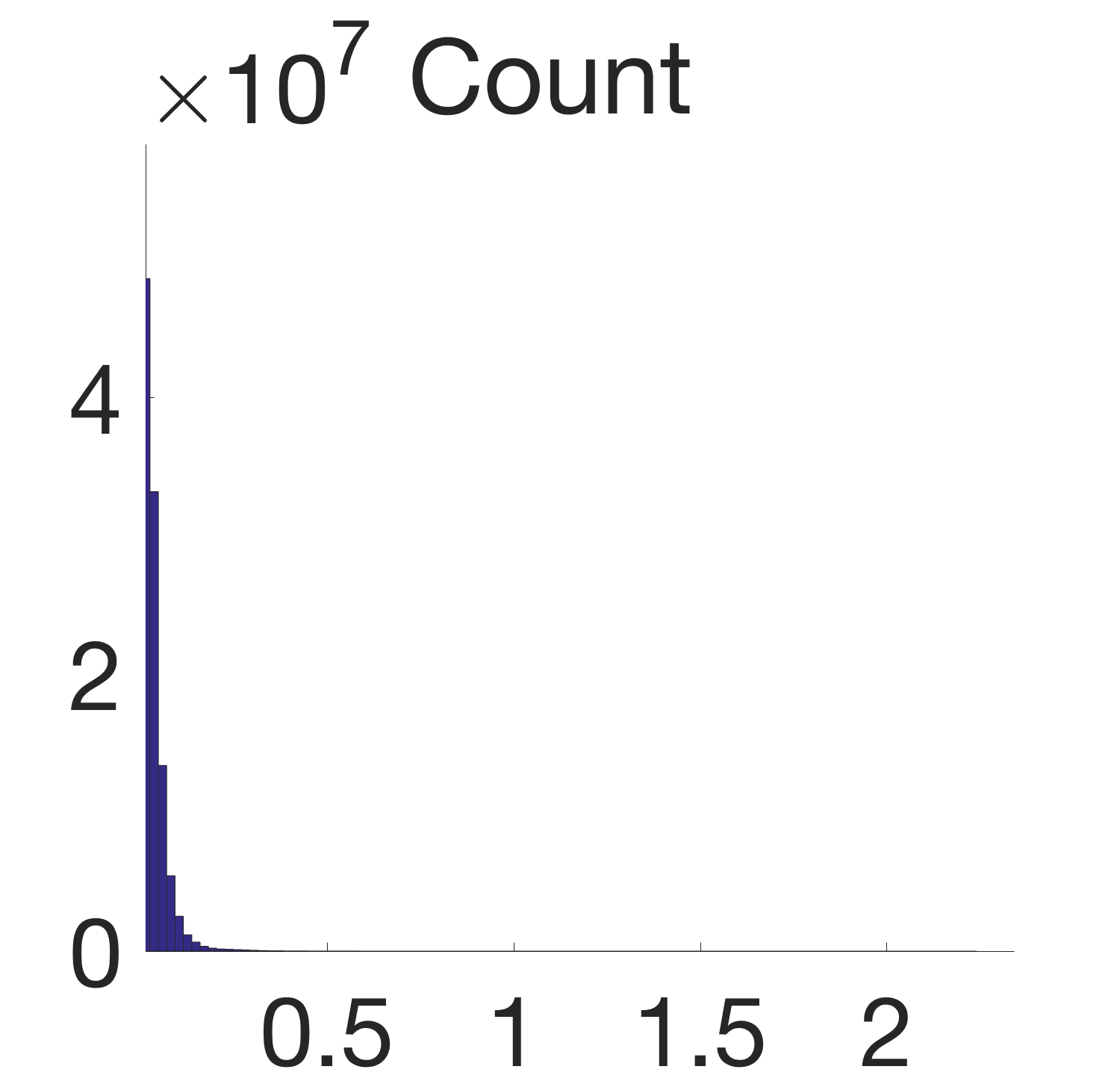}  
  \\
{\small (a) $\gamma = -0.5$, }  &   
{\small (b) $\gamma = 0$,}  &   
{\small (c) $\gamma = 0.5$, }  
\\
{\small skewness $ -1.92$. }  &
{\small skewness $0.15$. }  & 
{\small skewness $6.35$. }  
    \end{tabular}
  \end{center}
  \caption{\label{fig:hist} \emph{Symmetry matters.} A symmetric distribution of elements in $g^{-1} (\Pi)$ optimizes the empirical performance (BlogCatalog).  }
\end{figure}

\subsubsection{No memory vs. memory}
In node2vec, the memory term is adopted to design a biased random walk, which explores the graph structure in a controlled manner. The empirical results also show the superiority of node2vec over DeepWalk with appropriate memory parameters; however, the comparison suffers from the suboptimality of solving a non-convex optimization problem and randomness of obtaining biased random walks. Here we tune the memory parameters in various settings and obtain the corresponding closed-form solutions via singular value decomposition, which help us understand how the memory term influences the overall performance.  We consider GEM-D with the warping function $g(x) = \exp(x)$ and the FSMT matrix $\Pi^{(L,p,q)}$. We set $L=7$ and vary $p, q$ on the 2-D grid of  $[0.1, 0.5, 1,  2, 5]$.

We compare the classification performances on Kaggle 1968 in Figure{\color{red}~\ref{fig:kaggle_1968_node2vec}}, where the $x$-axis is $q$ and the $y$-axis is the F1 score. In each plot, five curves show the performance under various $p$. We set the labeling ratio to $50\%$. We see that the memory parameters influence the classification performance only slightly: under various settings of $p$ and $q$, the classification performances do not vary too much. The results show that the memory term can provide more flexibility and improves the overall performance, but it is not as significant as the walk length and nonlinearity.

\mypar{Summary} The memory factor does not matter  much.

\begin{figure}[thb]
  \begin{center}
    \begin{tabular}{cc}
\includegraphics[width=0.3\columnwidth]{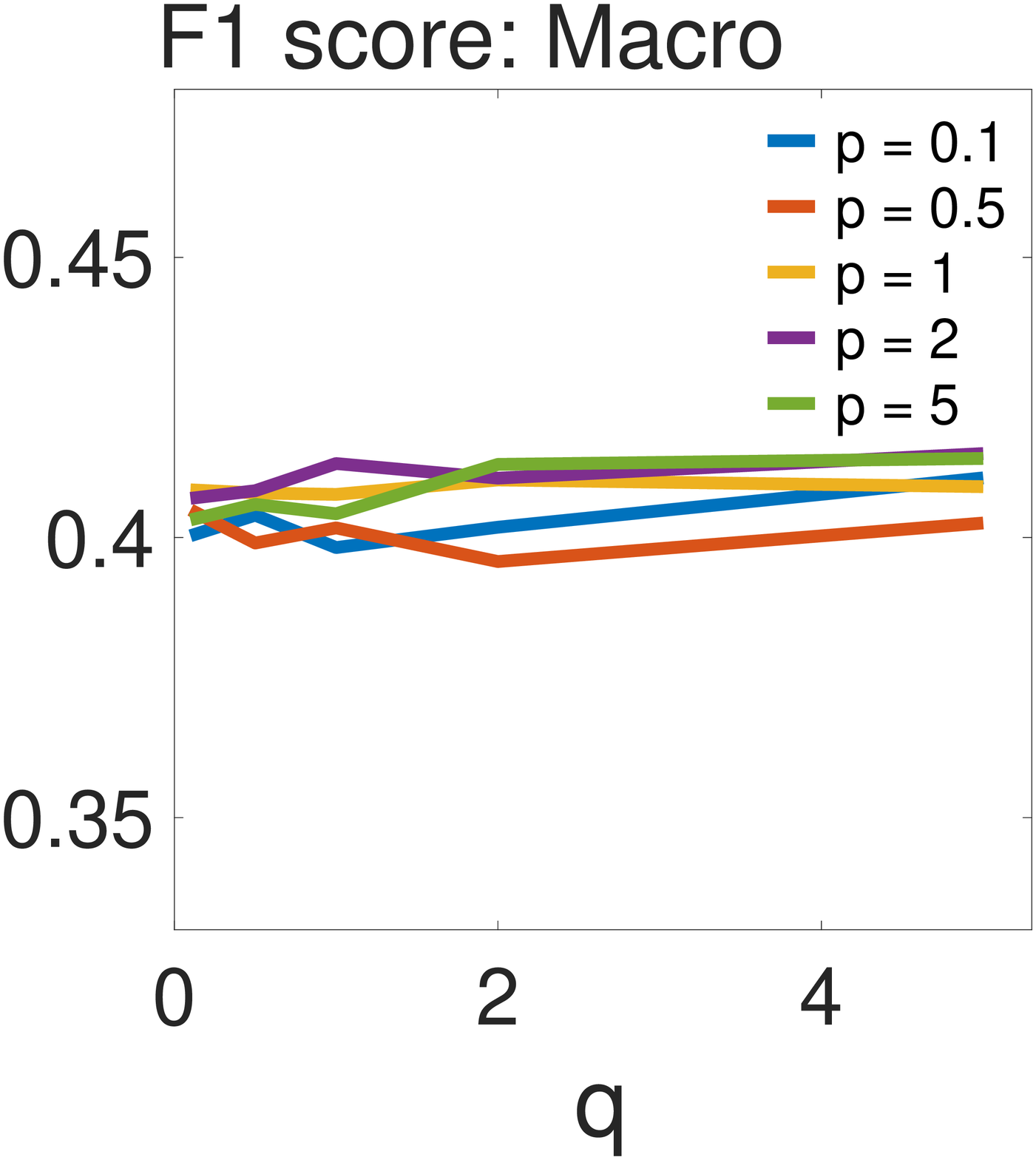}   & \includegraphics[width=0.3\columnwidth]{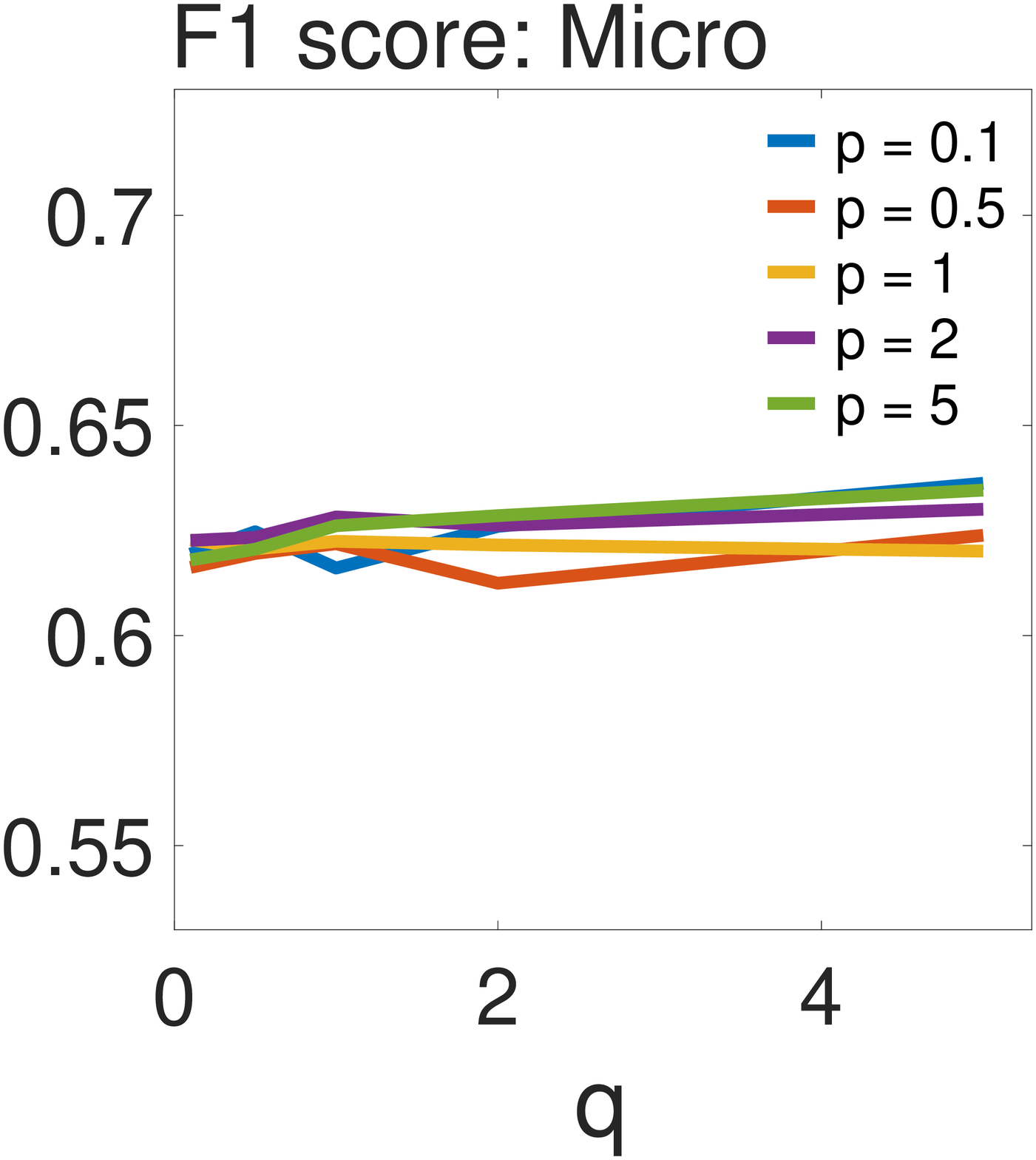} 
\\
 {\small (a) Macro at labeling ratio $50\%$. }  &   {\small (b) Micro at labeling ratio $50\%$. }
    \end{tabular}
  \end{center}
  \caption{\label{fig:kaggle_1968_node2vec}  \emph{Memory does not matter.} Memory factors do not significantly influence the classification performance (Kaggle 1968).}
\end{figure}



\subsection{Scalability}
We test scalable UltimateWalk on the U.S. patent 1975--1999.  We use the chronological order to construct a growing citation network (as new nodes and edges are added). We compute the corresponding graph embedding by using scalable UltimateWalk. We set the number of random walks $m = 50$. Figure{\color{red}~\ref{fig:patent}} shows its running time ($y$-axis) versus the number of input edges ($x$-axis). Red dots depict the graph at different years. The blue dotted line is not a linear fit, but approximately shows the linear trend. We see that scalable UltimateWalk scales linearly with the number of input edges, and embeds a 14-million edge graph under 40 minutes.

\section{Conclusions}
\label{sec:guide}
We conclude our work with a summary of our observations and suggest a practical guide to designing a graph embedding.

$\bullet$  {\bf Unifying framework.} GEM-D provides a unifying framework to design a graph embedding and subsumes previously proposed ones such as LapEigs, LINE, DeepWalk and node2vec; see Table{\color{red}~\ref{table:general}} and Section{\color{red}~\ref{sec:closed}}. GEM-D allows us to analyze the asymptotic behavior of DeepWalk and node2vec.

$\bullet$  {\bf Fundamental understanding.} We analyze each building block in GEM-D and observe that: (1) the warping function significantly influences the overall performance. The optimal warping function should appropiately rescale the distribution of the proximity matrix to be low-rank and symmetric. The exponential warping function usually maximizes the performance; see Figure{\color{red}~\ref{fig:blog_nonlinear}}; (2) the walk length in the node proximity function significantly influences the overall performance; both too large and too small values could be problematic. The empirical sweet spot is approximately the diameter of a graph; see Figure{\color{red}~\ref{fig:blog_step}}; and (3) the memory factor slightly influences the performance while introducing extra computational cost; see Figure{\color{red}~\ref{fig:kaggle_1968_node2vec}}.

$\bullet$ {\bf Fast, one-click algorithm.}  We propose UltimateWalk as a novel graph embedding algorithm, which is effective and without user-defined parameters and outperforms the competition (see Figure{\color{red}~\ref{fig:blog_converge}}), it scales linearly with the number of input edges (see Figure{\color{red}~\ref{fig:patent}} and Lemma{\color{red}~\ref{thm:complexity}}).

\bibliographystyle{IEEEbib}
\bibliography{sigproc} 

\end{document}